\documentclass[conference]{IEEEtran}
\IEEEoverridecommandlockouts
\usepackage{epsfig,psfrag}
\usepackage[cmex10]{amsmath}
\usepackage{amsthm,graphics,amscd,amssymb,enumerate,xspace,latexsym,
stmaryrd,wrapfig,subfig}
\usepackage{bbding}
\usepackage{calc}
\usepackage{pgf}
\usepackage{tikz}
\usetikzlibrary{arrows,automata,backgrounds,decorations}
\usetikzlibrary{shapes.multipart}
\usepgflibrary{decorations.pathreplacing}
\usetikzlibrary{arrows,calc,topaths}
\usetikzlibrary{automata,shapes.arrows, backgrounds, fadings,intersections, positioning}
\usetikzlibrary{shadows}
\tikzstyle{small}=[font=\footnotesize]
\tikzset{
    every picture/.style={>=stealth,auto,node distance=2cm},
}

\usetikzlibrary{petri}
\usetikzlibrary{shadows}
\usetikzlibrary{decorations.pathmorphing}
\usetikzlibrary{decorations.shapes}
\usetikzlibrary{shapes}
\usetikzlibrary{backgrounds}
\usetikzlibrary{fit}
\usetikzlibrary{arrows}
\usetikzlibrary{calc}
\usetikzlibrary{mindmap}
\usetikzlibrary{matrix}

\usepackage[capitalise,noabbrev,nameinlink]{cleveref}

\newcommand{\N}{\mathbb{N}}

\newcommand{\abs}[1]{\lvert#1\rvert}

\newcommand{\problemx}[3]{
\par\noindent\underline{\sc#1}\par\nobreak\vskip.2\baselineskip
\begingroup\clubpenalty10000\widowpenalty10000
\setbox0\hbox{\bf INPUT:\ }\setbox1\hbox{\bf QUESTION:\ }
\dimen0=\wd0\ifnum\wd1>\dimen0\dimen0=\wd1\fi
\vskip-\parskip\noindent
\hbox to\dimen0{\box0\hfil}\hangindent\dimen0\hangafter1\ignorespaces#2\par
\vskip-\parskip\noindent
\hbox to\dimen0{\box1\hfil}\hangindent\dimen0\hangafter1\ignorespaces#3\par
\endgroup}

\newcommand{\dist}{\mathcal{D}}
\newcommand{\supp}{{\sf supp}}

\newcommand{\buchi}[1]{\mathtt{B\ddot uchi}(#1)}
\newcommand{\reach}[1]{\mathtt{Reach}(#1)}
\newcommand{\reachn}[2]{\mathtt{Reach}_{#1}(#2)}
\newcommand{\reachp}[1]{\mathtt{Reach}^+(#1)}
\newcommand{\safety}[1]{\mathtt{Safety}(#1)}

\newtheorem{theorem}{Theorem}
\newtheorem{lemma}[theorem]{Lemma}

\newtheorem{remark}{Remark}

\newtheorem{example}{Example}

\newcommand{\hide}[1]{}

\newcommand{\lrc}[1]{(#1)}

\newcommand{\ignore}[1]{}

\newcommand{\emptyword}{\varepsilon}

\newcommand{\nat}{\mathbb N}
\newcommand{\ord}{\mathbb O}

\newcommand{\tuple}[1]{\lrc{#1}}

\newcommand{\game}{{\mathcal G}}

\newcommand{\gametuple}{\tuple{\states,(\zstates,\ostates,\rstates),\transition,\probp}}

\newcommand{\undef}{\bot}

\newcommand{\states}{S}

\newcommand{\state}{s}

\newcommand{\zstates}{\states_\zsymbol}
\newcommand{\ostates}{\states_\osymbol}
\newcommand{\rstates}{\states_\rsymbol}
\newcommand{\xstates}{\states_\xsymbol}
\newcommand{\zsymbol}{\Box}
\newcommand{\osymbol}{\Diamond}
\newcommand{\rsymbol}{\bigcirc}
\newcommand{\xsymbol}{\odot}

\newcommand{\transition}{{\longrightarrow}}

\newcommand{\probp}{P}

\newcommand{\play}{w}

\newcommand{\partialplay}{w}

\newcommand{\zstrat}{\sigma}
\newcommand{\ostrat}{\pi}
\newcommand{\xstrat}{\tau}

\newcommand{\zstratset}{\Sigma}
\newcommand{\ostratset}{\Pi}
\newcommand{\px}{\xsymbol}
\newcommand{\pz}{\zsymbol}
\newcommand{\po}{\osymbol}

\newcommand{\memory}{{\sf M}}
\newcommand{\memconf}{{\sf m}}
\newcommand{\memsuc}{\pi_s}
\newcommand{\memup}{\pi_u}
\newcommand{\memstrattuple}{\tuple{\memory,\initmem,\memup,\memsuc}}
\newcommand{\initmem}{\memconf_0}
\newcommand{\memstrat}[1]{{\sf T}^{#1}}
\newcommand{\memstratn}{\memstrat{}}

\newcommand{\probm}{{\mathcal P}}

\newcommand{\formula}{\mathcal{E}}

\newcommand{\zwinset}[2]{\big [#1 \big ]_{\pz}^{{#2}}}
\newcommand{\owinset}[2]{\big [#1 \big ]_{\po}^{{#2}}}

\newcommand{\reachset}{{\mathcal{T}\,\,\!\!}}

\newcommand{\valueof}[2]{{\mathtt{val}_{#1}(#2)}}

\newcommand{\constraint}{\rhd}
\newcommand{\nconstraint}{{\not\!\rhd}}
\newcommand{\const}{c}
\newcommand{\quantobj}[2]{({#1},{#2})}

\newcommand{\rvi}{{\it RVI}}

\begin{document}

\title{On Strong Determinacy of\\ Countable Stochastic Games}

\author{
\IEEEauthorblockN{Stefan Kiefer\IEEEauthorrefmark{1},
Richard Mayr\IEEEauthorrefmark{2},
Mahsa Shirmohammadi\IEEEauthorrefmark{1},
Dominik Wojtczak\IEEEauthorrefmark{3}}
\IEEEauthorblockA{\IEEEauthorrefmark{1}University of Oxford, UK}
\IEEEauthorblockA{\IEEEauthorrefmark{2}University of Edinburgh, UK}
\IEEEauthorblockA{\IEEEauthorrefmark{3}University of Liverpool, UK}
}

\IEEEoverridecommandlockouts
\IEEEpubid{\makebox[\columnwidth]{
Extended version of material presented at LICS 2017. arXiv.org - CC BY 4.0.
\hfill
} \hspace{\columnsep}\makebox[\columnwidth]{ }
}

\maketitle

\begin{abstract}
We study 2-player turn-based perfect-information 
stochastic games with countably infinite state space. 
The players aim at maximizing/minimizing the probability of a given event
(i.e., measurable set of infinite plays), 
such as reachability, B\"uchi, $\omega$-regular or more general objectives.

These games 
are known to be weakly determined, i.e., they have value.
However, strong determinacy of threshold objectives (given by an event~$\formula$ and a threshold $\const \in [0,1]$) was open in many cases:
is it always the case that the maximizer or the minimizer has a winning strategy, i.e., one that enforces, against all strategies of the other player, that 
$\formula$ is satisfied with
probability $\ge \const$ (resp.\ $< \const$)?

We show that almost-sure objectives  (where $c=1$)  are strongly determined.
This vastly generalizes a previous result on finite games with almost-sure tail objectives.
On the other hand we show that $\ge 1/2$ (co-)B\"uchi objectives 
are not strongly determined, not even if the game is finitely branching.

Moreover, for almost-sure reachability and almost-sure B\"uchi objectives in finitely branching games, we strengthen strong determinacy by showing that one of the players must have a memoryless deterministic (MD) winning strategy.
\end{abstract}

\begin{IEEEkeywords}
stochastic games, strong determinacy, infinite state space
\end{IEEEkeywords}

\section{Introduction}\label{sec:introduction}
{\bf\noindent Stochastic games.}
Two-player stochastic games \cite{Filar_Vrieze:book} are adversarial games between two 
players (the maximizer $\pz$ and the minimizer $\po$)
where some decisions are determined randomly according
to a pre-defined distribution.
Stochastic games are also called $2\frac{1}{2}$-player games 
in the terminology of \cite{Chatterjee:2004:QSP:982792.982808,chatterjee03simple}.
Player~$\pz$ tries to maximize the expected value of some payoff function
defined on the set of plays, while player~$\po$ tries to minimize it.
In concurrent stochastic games, in every round both players each
choose an action (out of given action sets) and for each combination of
actions the result is given by a pre-defined distribution.
In the subclass of turn-based stochastic games (also called simple stochastic
games) only one player gets to choose an action in every round, depending
on which player owns the current state.

\newcommand{\detmd}{\checkmark (MD)}
\newcommand{\newdetmd}{\CheckmarkBold \mbox{\bf(MD)}}
\newcommand{\newmddetmd}{\checkmark \mbox{\bf(MD)}}
\newcommand{\dnf}{\mbox{\checkmark ($\neg$FR)}}
\newcommand{\newdnf}{\CheckmarkBold ($\mathbf\neg${\bf FR})}
\newcommand{\newdetdnf}{\CheckmarkBold ($\neg${FR})}
\newcommand{\nd}{$\times$}
\newcommand{\newnd}{\XSolidBold}

\begin{table*}[ht]
\centering
\subfloat[Finitely branching games\label{taba}]{
\begin{tabular}{|l|c|c|c|c|}
\hline
\mbox{Objective}     & $>0$  & $>\const$ & $\ge\const$   & $=1$  \\ \hline
\mbox{Reachability}  & \detmd & \detmd  & \dnf          & \newmddetmd    \\  \hline
\mbox{B\"uchi}       & \newdetdnf  & \newnd  & \newnd   & \newdetmd    \\ \hline
\mbox{Borel}        &  \newdetdnf  & \newnd   & \newnd   & \newdetdnf  \\ \hline
\end{tabular}
}
\quad\quad
\subfloat[Infinitely branching games\label{tabb}]{
\begin{tabular}{|l|c|c|c|c|}
\hline
\mbox{Objective}     & $>0$  & $>\const$ & \hspace{3mm}$\ge\const$\hspace{3mm}   & $=1$  \\ \hline
\mbox{Reachability}  & \detmd    & \nd      & \nd           & \newdetdnf    \\  \hline
\mbox{B\"uchi}       & \newdetdnf  & \nd      & \nd           & \newdetdnf    \\ \hline
\mbox{Borel}        & \newdetdnf  & \nd      & \nd           & \newdetdnf  \\ \hline
\end{tabular}
}
\caption{Summary of determinacy and memory requirement properties for
reachability, B\"uchi and Borel objectives and various probability
thresholds. The results for safety
and co-B\"uchi are implicit, e.g., $>0$ B\"uchi is dual to to {$=1$}
co-B\"uchi. Similarly, $\quantobj{\mbox{Objective}}{>\const}$
is dual to $\quantobj{\neg\mbox{Objective}}{\ge\const}$.
The results hold for every constant $\const \in (0,1)$.
Tables~\ref{taba} and \ref{tabb} show the results for finitely branching and
infinitely branching countable games, respectively.
``\detmd'' stands for ``strongly MD-determined'', ``\dnf'' stands for 
``strongly determined but not strongly FR-determined'' and \nd\ stands
for ``not strongly determined''. 
New results are in boldface.
(All these objectives are weakly determined
by \cite{Maitra-Sudderth:1998}.)}
\label{tab:overview}
\end{table*}

We study 2-player turn-based perfect-information 
stochastic games with \emph{countably infinite} state spaces.
We consider objectives defined via predicates on plays,
not general payoff functions.
Thus the expected payoff value corresponds to the
probability that a play satisfies the predicate.

Standard questions are whether a game is determined,
and whether the strategies of the players can
without restriction be chosen to be of a particular 
type, e.g., MD (memoryless deterministic) or
FR (finite-memory randomized).

{\bf\noindent Finite-state games vs. Infinite-state games.}
Stochastic games with finite state spaces have been extensively studied
\cite{shapley-1953-stochastic,condon-1992-ic-complexity,AHK:FOCS98,GimbertHornSoda10,Chatterjee:2004:QSP:982792.982808},
both w.r.t.\ their determinacy and the strategy complexity (memory
requirements and randomization).
E.g., strategies in \emph{finite} stochastic parity games
can be chosen memoryless deterministic (MD) 
\cite{alfaro-2000-lics-concurrent,chatterjee03simple,chatterjee-2006-qest-strategy}.
These results have a strong influence
on algorithms for deciding the winner of stochastic games,
because such algorithms 
often use a structural property that the strategies can
be chosen of a particular type (e.g., MD or finite-memory).

More recently, several classes of finitely presented infinite-state games have been
considered as well. These are often induced by various types of automata
that use infinite memory (e.g., unbounded pushdown stacks, unbounded counters,
or unbounded fifo-queues). Most of these classes are still finitely branching.
Stochastic games on infinite-state probabilistic recursive systems (i.e.,
probabilistic pushdown automata with unbounded stacks) were studied in
\cite{Etessami:Yannakakis:ICALP05,EY:LMCS2008,EWY:ICALP08},
and stochastic games on systems with unbounded fifo-queues were
studied in \cite{ACMS:LMCS2014}.
However, most these works used techniques that are specially adapted to the
underlying automata model, not a general analysis of infinite-state games.  
Some results on general stochastic games with countably infinite state spaces
were presented in
\cite{Kucherabook11,BBKO:IC2011,Krcal:Thesis:2009,Brozek:TCS2013}
though many questions remained open (see our contributions further below).

It should be noted that many standard results and proof techniques
from finite games do \emph{not}
carry over to countably infinite games.
E.g., 
\begin{itemize}
\item
Even if a state has value, an optimal strategy need not exist, 
not even for reachability
objectives \cite{Kucherabook11}.
\item
Some strong determinacy properties (see below) do not hold,
not even for reachability objectives \cite{BBKO:IC2011,Krcal:Thesis:2009}
(while in finite games they hold even for parity objectives
\cite{Chatterjee:2004:QSP:982792.982808}).
\item
The memory requirements of optimal strategies
are different. In finite games, optimal strategies for
parity objectives can be chosen memoryless deterministic
\cite{Chatterjee:2004:QSP:982792.982808}.
In contrast, in countably infinite games (even if finitely branching)
optimal strategies for reachability objectives, where they exist, require
infinite memory \cite{Kucherabook11}.
\end{itemize}
One of the reasons underlying this difference is the following.
Consider the values of the states in a game w.r.t.\ a certain objective.
If the game is finite then there are only finitely many such values,
and in particular there exists some minimal nonzero value (unless all states
have value zero).
This property does \emph{not carry over} to infinite games. 
Here the set of states is infinite and the infimum over the nonzero values can be zero.
As a consequence, even for a reachability objective, it is possible that all states have 
value $>0$, but still the value of some states is $<1$. Such phenomena appear
already in infinite-state Markov chains like the classic Gambler's ruin problem with unfair
coin tosses in the player's favor (e.g., $0.6$ win and $0.4$ lose). The value,
i.e., the probability of ruin, is always $>0$, but still $<1$ in every state
except the ruin state itself; cf.~\cite{Feller:book} (Chapt.~14).

\smallskip
{\bf\noindent Weak determinacy.}
Using Martin's result \cite{Martin:1998},
Maitra \& Sudderth \cite{Maitra-Sudderth:1998} showed that stochastic games with Borel
payoffs are \emph{weakly determined}, i.e., all states have value.
This very general result holds even for concurrent games and general (not necessarily countable) 
state spaces. They work in the framework of finitely additive probability
theory (under weak assumptions on measures) and only assume a finitely additive law of motion.
Also their payoff functions are general bounded Borel measurable functions, 
not necessarily predicates on plays.

\smallskip
{\bf\noindent Strong determinacy.}
Given a predicate $\formula$ on plays and a constant $\const \in [0,1]$, 
strong determinacy 
of a threshold objective $\quantobj{\formula}{\constraint\const}$
(where $\constraint\in\{>,\ge\}$)
holds iff either the maximizer or the minimizer
has a winning strategy, i.e., a strategy that enforces (against any 
strategy of the other player) that 
the predicate $\formula$ holds with
probability $\constraint \const$ (resp.\ $\mathrel{\not\constraint} \const$).
In the case of $\quantobj{\formula}{=1}$, 
one speaks of an almost-sure $\formula$ objective.
If the winning strategy of the winning player 
can be chosen MD (memoryless deterministic)
then one says that the threshold objective is strongly 
MD determined. Similarly for other types of strategies, e.g.,
FR (finite-memory randomized).

\smallskip
{\bf\noindent Strong determinacy in finite games.}
Strong determinacy for almost-sure objectives $\quantobj{\formula}{=1}$ 
(and for the dual positive probability objectives $\quantobj{\formula}{>0}$) 
is sometimes called \emph{qualitative determinacy}~\cite{GimbertHornSoda10}.
In~\cite[Theorem~3.3]{GimbertHornSoda10} it is shown that \emph{finite} 
stochastic games with Borel \emph{tail} (i.e., prefix-independent)
objectives are qualitatively determined.
(We'll show a more general result for countably infinite games and general 
objectives; see below.)
In the special case of parity objectives,
even strong MD determinacy holds for any threshold
$\constraint\const$ \cite{Chatterjee:2004:QSP:982792.982808}.

\smallskip
{\bf\noindent Strong determinacy in infinite games.}
It was shown in~\cite{BBKO:IC2011,Krcal:Thesis:2009,Brozek:TCS2013}
that in finitely branching games with countable state spaces 
reachability objectives with any threshold
$\constraint\const$ with $\const \in [0,1]$, are strongly determined.
However, the player~$\pz$ strategy may need infinite
memory \cite{Kucherabook11},
and thus reachability objectives are not strongly MD determined. 
Strong determinacy 
does \emph{not} hold for infinitely branching reachability games 
with thresholds $\constraint\const$ with $\const \in (0,1)$; cf.\ Figure 1 in \cite{BBKO:IC2011}.

\smallskip
{\bf\noindent Our contribution to determinacy.}
We show that almost-sure Borel objectives are strongly determined
for games with \emph{countably infinite} state spaces.
(In particular this even holds for infinitely branching games; cf.~Table~\ref{tab:overview}.)
This removes both the restriction to finite games and the
restriction to tail objectives of \cite[Theorem~3.3]{GimbertHornSoda10},
and solves an open problem stated there.
(To the best of our knowledge, strong determinacy was open even
for almost-sure reachability objectives in infinitely branching
countable games.)

On the other hand, we show that, for countable games,
$\constraint\const$ (co-)B\"uchi objectives 
are not strongly determined for any $\const \in (0,1)$, 
not even if the game graph is finitely branching.

\smallskip
{\bf\noindent Our contribution to strategy complexity.}
While $\constraint\const$ reachability objectives in 
finitely branching countable games are not strongly MD determined 
in general \cite{Kucherabook11},
we show that strong MD determinacy holds for many interesting subclasses.
In finitely branching games, 
it holds for strict inequality $>\const$ reachability, almost-sure
reachability, and in all games where either player~$\pz$ does not have any
value-decreasing transitions or player~$\po$ does not have any
value-increasing transitions.

Moreover, we show that almost-sure B\"uchi objectives
(but not almost-sure co-B\"uchi objectives) are strongly MD determined,
provided that the game is finitely branching.

Table~\ref{tab:overview} summarizes all properties of strong determinacy and memory
requirements for Borel objectives and subclasses on countably infinite games.

\section{Preliminaries}\label{sec:preliminaries}
A \emph{probability distribution} over a countable (not necessarily finite) set~$S$ is a function
$f:\states\to[0,1]$ s.t.~$\sum_{\state\in \states}f(\state)=1$.
We use~$\supp(f) = \{\state \in \states \mid f(\state) > 0\}$ to denote the \emph{support} of~$f$.
Let $\dist(\states)$ be the set of all probability distributions over~$\states$.

We consider $2\frac{1}{2}$-player games where players have perfect information and play in turn for infinitely many rounds.
\emph{Games}~$\game=\gametuple$ are defined such that the countable set of \emph{states}  is partitioned into
the  set~$\zstates$ of states  of player$~\pz$, the set~$\ostates$ of states of player~$\po$ and
\emph{random states} $\rstates$.
The relation $\mathord{\transition}\subseteq\states\times\states$ is the transition relation.
We write $\state\transition{}\state'$ if $\tuple{\state,\state'}\in \transition$, and we assume that each state~$\state$ has a \emph{successor} state $\state'$ with $\state \transition \state'$.
The probability function $\probp:\rstates \to \dist(\states)$ assigns to each random 
state~$\state \in \rstates$ a probability distribution over its 
successor states. 
%
%
The game~$\game$ is called \emph{finitely branching} if each state has only finitely many successors;
otherwise, it is \emph{infinitely branching}.
Let $\xsymbol \in \{\zsymbol,\osymbol\}$.
If $\xstates=\emptyset$, we say that player~$\px$ is \emph{passive}, and the game is a \emph{Markov decision process (MDP)}. 
%
A Markov chain is an MDP where both players are passive. 

The stochastic game is played by two players $\pz$ (maximizer) and $\po$ (minimizer).
The game starts  in a given  initial state~$\state_0$ and evolves for infinitely many rounds.
In each round,  if the game is in state $\state \in \xstates$ 
then player $\xsymbol$ chooses a successor state~$s'$ with~$s\transition{}s'$;
otherwise the game is in a random state~$s\in \rstates$ and proceeds randomly to~$s'$ with probability~$\probp(s)(s')$.

\medskip

\noindent {\bf Strategies.}
A \emph{play}~$\play$  is an infinite sequence
$\state_0\state_1\cdots \in \states^\omega$ of states 
such that  $\state_i\transition{}\state_{i+1}$ for all $i\geq 0$;
let  $\play(i)=\state_i$ denote the $i$-th state along~$\play$.
A \emph{partial play} is a finite prefix of a play.
We say that (partial) play $\play$ \emph{visits} $\state$ if
$\state=\play(i)$ for some $i$, and that~$\play$ starts in~$s$ if $\state=\play(0)$. 
A \emph{strategy} of the player~$\pz$ is a 
function $\zstrat:\states^*\zstates \to \dist(S)$ that assigns to partial plays 
$\partialplay\state \in \states^*\zstates$ 
a distribution over the successors~$\{\state'\in \states\mid \state \transition{} \state'\}$. 
Strategies~$\ostrat:\states^*\ostates \to \dist(S)$ for the player~$\po$ are defined analogously.
The set of all strategies of player~$\pz$ and player~$\po$  in 
$\game$ is denoted by $\zstratset_\game$ and $\ostratset_\game$,
respectively  (we omit the subscript and write~$\zstratset$ and $\ostratset$ if $\game$ is clear).
A (partial) play~$\state_0\state_1\cdots$ is induced by strategies~$\tuple{\zstrat,\ostrat}$
if~$\state_{i+1}\in \supp(\zstrat(\state_0\state_1\cdots\state_i))$ for all~$\state_i \in \zstates$, 
and if~$\state_{i+1} \in\supp(\ostrat(\state_0\state_1\cdots\state_i))$ for all~$\state_i \in \ostates$.
%

To emphasize the amount of memory required to implement a strategy,
we present an equivalent formulation of strategies.
A strategy of player~$\xsymbol$ can be implemented by a probabilistic transducer~$\memstratn=\memstrattuple$
where $\memory$ is a  countable set (the memory of the strategy),
$\memconf_0\in\memory$ is the initial memory mode
and  $\states$ is the input and output alphabet.
The probabilistic transition function~$\memup : \memory \times \states\to \dist(\memory)$ updates the 
memory mode of the transducer. 
The probabilistic successor function~$\memsuc : \memory \times \xstates \to \dist(\states)$ 
outputs the next successor, where $\state' \in \supp(\memsuc(\memconf,\state))$ implies 
$\state\transition{}\state'$.
We extend $\memup$ to $\dist(\memory) \times \states\to \dist(\memory)$ and  
$\memsuc$ to $\dist(\memory) \times \xstates\to \dist(\states)$, in the natural way.
Moreover, we extend $\memup$
to paths  by
$\memup(\memconf,\emptyword)=\memconf$ and
$\memup(\memconf,\state_0\cdots\state_n) =
\memup(\memup(\state_0\cdots\state_{n-1},\memconf),\state_n)$.
The strategy $\xstrat_{\memstratn}:\states^*\xstates\to\dist(\states)$
induced by the transducer~$\memstratn$ is given by 
$\xstrat_{\memstratn}(\state_0\cdots\state_{n}):=\memsuc(\state_{n},\memup(\state_0\cdots\state_{n-1},\initmem))$.

Strategies are in general \emph{history dependent}~(H) and \emph{randomized}~(R).
An H-strategy~$\xstrat \in\ \{\zstrat,\ostrat\}$ is \emph{finite memory}~(F) if there exists some transducer $\memstratn$ with memory~$\memory$
such that  $\xstrat_{\memstratn}=\xstrat$ and $\abs{\memory}<\infty$; otherwise 
$\xstrat$  \emph{requires infinite memory}.
An F-strategy is   \emph{memoryless}~(M) (also called \emph{positional}) 
if~$\abs{\memory}=1$. 
For convenience, we may view 
M-strategies as functions $\xstrat: \xstates \to \dist(\states)$.
An R-strategy $\xstrat$ is \emph{deterministic}~(D)
if $\memup$ and $\memsuc$ map to Dirac distributions;
it implies that $\xstrat(\partialplay)$ is a Dirac distribution for all partial plays~$\partialplay$.
All combinations of the properties in $\{\text{M},\text{F},\text{H}\}\times \{\text{D},\text{R}\}$ are possible, e.g., MD stands for
memoryless deterministic. HR strategies are the most general type.

\medskip

\noindent {\bf Probability Measure and Events.} 
To a game~$\game$, an initial state~$\state_0$ and strategies~$\tuple{\zstrat,\ostrat}$
we associate the standard probability
space~$(\state_0\states^\omega,\mathcal{F},\probm_{\game,\state_0,\zstrat,\ostrat})$
w.r.t.\ the induced Markov chain.
First one defines a topological space on the set of infinite
plays~$\state_0 \states^{\omega}$. The 
\emph{cylinder sets} are the sets $s_0 s_1 \ldots s_n \states^\omega$, where
$s_1, \ldots, s_n \in \states$
and the open sets are arbitrary unions of cylinder sets, i.e., the sets
$Y \states^{\omega}$ with $Y\subseteq s_0 \states^{*}$.
The Borel $\sigma$-algebra 
${\mathcal F} \subseteq 2^{\state_0 \states^\omega}$ is the smallest
$\sigma$-algebra that contains all the open sets.

The probability measure $\probm_{\game,\state_0,\zstrat,\ostrat}$ is obtained
by first defining it on the cylinder sets and then extending it to all sets in 
the Borel $\sigma$-algebra.
If $s_0 s_1 \ldots s_n$ is not a partial play induced by~$\tuple{\zstrat,\ostrat}$ then let 
$\probm_{\game,\state_0,\zstrat,\ostrat}(s_0 s_1 \ldots s_n \states^\omega) = 0$; 
otherwise let $\probm_{\game,\state_0,\zstrat,\ostrat}(s_0 s_1 \ldots s_n \states^\omega) = \prod_{i=0}^{n-1} \xstrat(s_0 s_1 \ldots s_i)(s_{i+1})$, where $\xstrat$ is such that
$\xstrat(ws)=\zstrat(ws)$ for all $w s \in \states^* \zstates$,
$\xstrat(ws)=\ostrat(ws)$ for all $w s \in \states^* \ostates$, and 
$\xstrat(w s) = \probp(s)$ for all $w s \in \states^* \rstates$.
By Carath\'eodory's extension theorem~\cite{billingsley-1995-probability}, 
this defines a unique probability
measure~$\probm_{\game,\state_0,\zstrat,\ostrat}$ on the Borel
$\sigma$-algebra $\mathcal{F}$.

We will call any set $\formula \in \mathcal{F}$ an \emph{event}, i.e., 
an event is a measurable (in the probability space above) set of infinite plays.
Equivalently, one may view an event~$\formula$ as a Borel measurable payoff
function of the form $\formula: \state_0 \states^\omega \to \{0,1\}$.
Given $\formula' \subseteq \states^\omega$ (where potentially $\formula' \not\subseteq \state_0 \states^\omega$) we often write $\probm_{\game,\state_0,\zstrat,\ostrat}(\formula')$ for $\probm_{\game,\state_0,\zstrat,\ostrat}(\formula' \cap \state_0 \states^\omega)$ to avoid clutter.
 
\medskip
\noindent {\bf Objectives.} Let $\game = \gametuple$ be a game.
The objectives of the players are determined by  events~$\formula$.
We write~$\neg \formula$ for the dual  objective defined as
$\neg \formula=S^{\omega} \setminus \formula$.

Given a target set $\reachset \subseteq \states$, 
the \emph{reachability objective} is defined by 
the event $$\reach{\reachset}=\{\state_0\state_1\cdots \in S^{\omega}\mid \exists i.\,
\state_i \in \reachset\}.$$
Moreover, $\reachn{n}{\reachset}$ denotes the set of all 
plays visiting~$\reachset$ in the first $n$ steps, i.e.,
$\reachn{n}{\reachset} =\{\state_0\state_1\cdots\mid \exists
i\le n.\, \state_i \in \reachset\}$.
The \emph{safety objective} is defined as the dual of reachability:
$\safety{\reachset}=\neg \reach{\reachset}$.

For a set $\reachset \subseteq \states$ of states called \emph{B\"uchi states}, 
the \emph{B\"uchi objective} is the event 
$$\buchi{\reachset}=\{\state_0\state_1\cdots\in S^{\omega}\mid \forall i \,\exists j\geq i.\,
\state_j\in \reachset\}.$$
The \emph{co-B\"uchi objective} is defined as the dual of B\"uchi.

%

Note that the objectives of player~$\pz$ (maximizer) and player~$\po$
(minimizer) are dual to each other. Where player~$\pz$ tries
to maximize the probability of some objective $\formula$,
player~$\po$ tries to maximize the probability of $\neg\formula$.

\section{Determinacy}\label{sec:determinacy}
\subsection{Optimal and $\epsilon$-Optimal Strategies; Weak and Strong Determinacy}

Given an objective~$\formula$ for player~$\zsymbol$ in a game $\game$,
 state $\state$ has \emph{value} if  
\[\sup_{\zstrat\in\zstratset}\inf_{\ostrat\in\ostratset}\probm_{\game,\state,\zstrat,\ostrat}(\formula)
= \inf_{\ostrat\in\ostratset}\sup_{\zstrat\in\zstratset}\probm_{\game,\state,\zstrat,\ostrat}(\formula).\]

If $\state$ has value then $\valueof{\game}{\state}$ denotes the value of
$\state$ defined by the above equality. 
A game with a fixed objective is called \emph{weakly determined}
iff every state has value.

\begin{theorem}[follows immediately from~\cite{Maitra-Sudderth:1998}]\label{thm:weak_borel_determinacy}
Countable stochastic games (as defined in Section~\ref{sec:preliminaries}) are weakly determined.
\end{theorem}
Theorem~\ref{thm:weak_borel_determinacy} is an immediate consequence of a far
more general result by Maitra \& Sudderth~\cite{Maitra-Sudderth:1998} on weak determinacy of (finitely additive) games with general Borel payoff objectives.

For $\epsilon \ge 0$ and $\state \in \states$, we say that
\begin{itemize}
\item
$\zstrat \in \zstratset$ is \emph{$\epsilon$-optimal (maximizing)}  iff
$\probm_{\game,\state,\zstrat,\ostrat}(\formula) \ge \valueof{\game}{\state}
-\epsilon$ for all $\ostrat \in \ostratset$.
\item
$\ostrat \in \ostratset$ is \emph{$\epsilon$-optimal (minimizing)} iff
$\probm_{\game,\state,\zstrat,\ostrat}(\formula) \le \valueof{\game}{\state}
+\epsilon$ for all $\zstrat \in \zstratset$.
\end{itemize}
A $0$-optimal strategy is called \emph{optimal}.
An optimal strategy for the player~$\zsymbol$ is \emph{almost-surely} winning if $\valueof{\game}{\state}=1$.
Unlike in finite-state games, optimal strategies need not exist
in countable games, not even for reachability objectives in finitely
branching MDPs \cite{BBEKW10,BBKO:IC2011}.

However, since our games are weakly determined by Theorem~\ref{thm:weak_borel_determinacy}, 
for all $\epsilon>0$ there exist $\epsilon$-optimal strategies for both players.

For an objective~$\formula$ and $\constraint\in\{\mathord{\geq},\mathord{>}\}$ and threshold $\const \in [0,1]$,
we define \emph{threshold objectives} $\quantobj{\formula}{\constraint\const}$
as follows.
\begin{itemize}
	\item  ${\zwinset{\formula}{\constraint\const}}_{\!\!\game}$ is the set of states~$\state$ for which there exists 
a strategy~$\zstrat$ such that, for all $\ostrat \in \ostratset$,  we have $\probm_{\game,\state,\zstrat,\ostrat}(\formula) \constraint \const$.
\item ${\owinset{\formula}{\nconstraint\const}}_\game$ is the set of states~$\state$ for which there exists 
a strategy~$\ostrat$ such that, for all $\zstrat \in \zstratset$,  we have $\probm_{\game,\state,\zstrat,\ostrat}(\formula) \nconstraint \const$.
\end{itemize}
We omit the subscript $\game$ where it is clear from the context.
We call a state~$\state$ \emph{almost-surely winning} for the player~$\pz$ iff $\state \in \zwinset{\formula}{\ge 1}$.

By the duality of the players, a $\quantobj{\formula}{\ge\const}$
objective for player~$\pz$ corresponds to a
$\quantobj{\neg\formula}{>1-\const}$
objective from player~$\po$'s point of view. 
E.g., an almost-sure B\"uchi objective for player~$\pz$ 
corresponds to a positive-probability 
co-B\"uchi objective for player~$\po$.
Thus we can restrict our attention to reachability, B\"uchi and general (Borel set) objectives,
since safety is dual to reachability, and co-B\"uchi is dual 
to B\"uchi, and Borel is self-dual.

A game $\game$ with 
threshold objective $\quantobj{\formula}{\constraint\const}$
is called \emph{strongly determined} iff
in every state~$\state$ either player~$\pz$ or player~$\po$
has a winning strategy, i.e., iff
$\states = \zwinset{\formula}{\constraint\const} \uplus 
\owinset{\formula}{\nconstraint\const}$.

Strong determinacy depends on the specified threshold $\constraint\const$.
Strong determinacy for almost-sure objectives $\quantobj{\formula}{=1}$ (and for the dual positive probability objectives $\quantobj{\formula}{>0}$) is sometimes called \emph{qualitative determinacy}~\cite{GimbertHornSoda10}.
In~\cite[Theorem~3.3]{GimbertHornSoda10} it is shown that \emph{finite} stochastic games with \emph{tail} objectives are qualitatively determined.
An objective~$\formula$ is called \emph{tail} if for all $w_0 \in \states^*$ and all $w \in \states^\omega$ we have $w_0 w \in \formula \Leftrightarrow w \in \formula$, i.e., a tail objective is independent of finite prefixes.
The authors of~\cite{GimbertHornSoda10} express ``hope that [their qualitative determinacy theorem] may be extended beyond the class of finite simple stochastic tail games''.
We fulfill this hope by generalizing their theorem from finite to countable games and from tail objectives to arbitrary objectives:

\begin{theorem}\label{thm:amost-sure-strong-det}
Stochastic games, even infinitely branching ones, with almost-sure objectives 
are strongly determined.
\end{theorem}
Theorem~\ref{thm:amost-sure-strong-det} does not carry over to thresholds other than 0~or~1; cf.~Theorem~\ref{thm:Buchi-no-strong-det}.

The main ingredients of the proof of Theorem~\ref{thm:amost-sure-strong-det} are transfinite induction, weak determinacy of stochastic games (Theorem~\ref{thm:weak_borel_determinacy}), the concept of a ``reset'' strategy from~\cite{GimbertHornSoda10}, and L\'evy's zero-one law.
The principal idea of the proof is to construct a transfinite sequence of subgames, by removing parts of the game that player~$\pz$ cannot risk entering.
This approach is used later in this paper as well, for Theorems \ref{thm:reachability} and~\ref{thm:Buchi}.

\begin{example} \label{ex:exampleLimitOrdinal}
We explain this approach using the reachability game in Figure~\ref{fig:exampleLimitOrdinal} as an example.
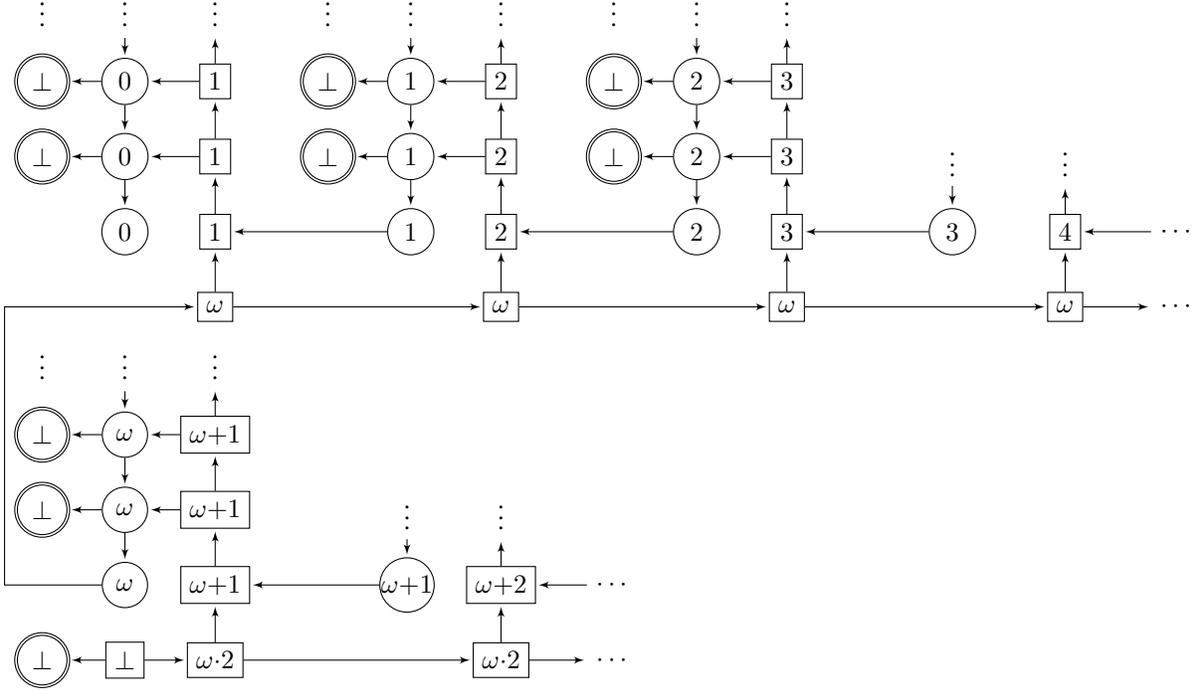
\begin{figure*}[t]
\centering
\begin{tikzpicture}[>=latex',shorten >=1pt,node distance=1.9cm,on grid,auto,
roundnode/.style={circle, draw,minimum size=1.5mm},
squarenode/.style={rectangle, draw,minimum size=2mm},
diamonddnode/.style={diamond,draw, minimum size=8mm, inner sep=-10pt}]

\node [roundnode,double] (r0) at(0,4.7) [draw]{$\bot$};
\node [roundnode,double] (r1) [draw,above=1cm of r0]{$\bot$};
\node [roundnode,double] (r2) [draw=none,above=1cm of r1]{$\vdots$};

\node [roundnode] (s0) [draw,right=1.1cm of r0] {$0$};
\node [roundnode] (s1) [draw,above=1cm of s0] {$0$};
\node [roundnode] (s2) [draw=none,above=1cm of s1]{$\vdots$}; 

\node [roundnode] (s4) [draw,below=1cm of s0]{$0$};

\node [squarenode] (ss0) [draw,right=1.2cm of s0] {$1$};
\node [squarenode] (ss1) [draw,above=1cm of ss0] {$1$};
\node [squarenode] (ss2) [draw=none,above=1cm of ss1]{$\vdots$}; 

\node [squarenode] (ss4) [draw,below=1cm of ss0]{$1$};

\path[->] (s0) edge (r0);
\path[->] (s1) edge (r1);

\path[->] (ss0) edge (s0);
\path[->] (ss1) edge (s1);

\path[->] (s2) edge (s1);
\path[->] (s1) edge (s0);
\path[->] (s0) edge (s4);

\path[->] (ss4) edge (ss0);
\path[->] (ss0) edge (ss1);
\path[->] (ss1) edge (ss2);

\node [squarenode] (t1) [draw,below=1cm of ss4]{$\omega$};
\path[->] (t1) edge (ss4);

\node [roundnode,double] (r0) [draw,right=1.5cm of ss0]{$\bot$};
\node [roundnode,double] (r1) [draw,above=1cm of r0]{$\bot$};
\node [roundnode,double] (r2) [draw=none,above=1cm of r1]{$\vdots$};

\node [roundnode] (s0) [draw,right=1.1cm of r0] {$1$};

\node [roundnode] (s1) [draw,above=1cm of s0] {$1$};
\node [roundnode] (s2) [draw=none,above=1cm of s1]{$\vdots$}; 

\node [roundnode] (s4) [draw,below=1cm of s0]{$1$};
\path[->] (s4) edge (ss4);

\node [squarenode] (ss0) [draw,right=1.2cm of s0] {$2$};
\node [squarenode] (ss1) [draw,above=1cm of ss0] {$2$};
\node [squarenode] (ss2) [draw=none,above=1cm of ss1]{$\vdots$}; 

\node [squarenode] (ss4) [draw,below=1cm of ss0]{$2$};

\path[->] (s0) edge (r0);
\path[->] (s1) edge (r1);

\path[->] (ss0) edge (s0);
\path[->] (ss1) edge (s1);

\path[->] (s2) edge (s1);
\path[->] (s1) edge (s0);
\path[->] (s0) edge (s4);

\path[->] (ss4) edge (ss0);
\path[->] (ss0) edge (ss1);
\path[->] (ss1) edge (ss2);

\node [squarenode] (t2) [draw,below=1cm of ss4]{$\omega$};
\path[->] (t2) edge (ss4);
\path[->] (t1) edge (t2);

\node [roundnode,double] (r0) [draw,right=1.5cm of ss0]{$\bot$};
\node [roundnode,double] (r1) [draw,above=1cm of r0]{$\bot$};
\node [roundnode,double] (r2) [draw=none,above=1cm of r1]{$\vdots$};

\node [roundnode] (s0) [draw,right=1.1cm of r0] {$2$};
\node [roundnode] (s1) [draw,above=1cm of s0] {$2$};
\node [roundnode] (s2) [draw=none,above=1cm of s1]{$\vdots$}; 

\node [roundnode] (s4) [draw,below=1cm of s0]{$2$};
\path[->] (s4) edge (ss4);

\node [squarenode] (ss0) [draw,right=1.2cm of s0] {$3$};
\node [squarenode] (ss1) [draw,above=1cm of ss0] {$3$};
\node [squarenode] (ss2) [draw=none,above=1cm of ss1]{$\vdots$};

\node [squarenode] (ss4) [draw,below=1cm of ss0]{$3$};
\node [roundnode] (m4) [draw,right=2.2cm of ss4]{$3$};
\path[->] (m4) edge (ss4);

\path[->] (s0) edge (r0);
\path[->] (s1) edge (r1);

\path[->] (ss0) edge (s0);
\path[->] (ss1) edge (s1);

\path[->] (s2) edge (s1);
\path[->] (s1) edge (s0);
\path[->] (s0) edge (s4);

\path[->] (ss4) edge (ss0);
\path[->] (ss0) edge (ss1);
\path[->] (ss1) edge (ss2);

\node [squarenode] (t3) [draw,below=1cm of ss4]{$\omega$};
\path[->] (t3) edge (ss4);
\path[->] (t2) edge (t3);

\node [squarenode] (dum0) [right=3.7cm of t3]{$\omega$};
\node [squarenode] (dum1) [ above=1cm of dum0]{$4$};
\node [squarenode] (dum6) [draw=none,right=1.5cm of dum1]{$\cdots$};
\node [squarenode] (dum3) [draw=none, above=1cm of dum1]{$\vdots$};
\path[->] (dum0) edge (dum1);
\path[->] (dum1) edge (dum3);
\path[->] (t3) edge (dum0);
\node [squarenode] (dum2) [draw=none,right=1.5cm of dum0]{$\cdots$};

\path[->] (dum0) edge (dum2);
\path[->] (dum6) edge (dum1);

\node [squarenode] (dum4) [draw=none, above=1cm of m4]{$\vdots$};
\path[->] (dum4) edge (m4);

\node [roundnode,double] (r0) at(0,0) [draw]{$\bot$};
\node [roundnode,double] (r1) [draw,above=1cm of r0]{$\bot$};
\node [roundnode,double] (r2)[draw=none,above=1cm of r1]{$\vdots$};

\node [roundnode] (s0) [draw,right=1.1cm of r0] {$\omega$};
\node [roundnode] (s1) [draw,above=1cm of s0] {$\omega$};
\node [roundnode] (s2) [draw=none,above=1cm of s1]{$\vdots$};

\node [roundnode] (s4) [draw,below=1cm of s0]{$\omega$};

\node [squarenode] (ss0) [draw,right=1.2cm of s0] {$\omega {+} 1$};
\node [squarenode] (ss1) [draw,above=1cm of ss0] {$\omega {+} 1$};
\node [squarenode] (ss2) [draw=none,above=1cm of ss1]{$\vdots$};

\node [squarenode] (ss4) [draw,below=1cm of ss0]{$\omega {+} 1$};
\node [roundnode,inner sep=0pt] (m4) [draw,right=2.55cm of ss4]{$\omega {+} 1$};
\path[->] (m4) edge (ss4);

\path[->] (s0) edge (r0);
\path[->] (s1) edge (r1);

\path[->] (ss0) edge (s0);
\path[->] (ss1) edge (s1);

\path[->] (s2) edge (s1);
\path[->] (s1) edge (s0);
\path[->] (s0) edge (s4);

\path[->] (ss4) edge (ss0);
\path[->] (ss0) edge (ss1);
\path[->] (ss1) edge (ss2);

\node [squarenode] (tt1) [draw,below=1cm of ss4]{$\omega {\cdot} 2$};
\path[->] (tt1) edge (ss4);

\draw [->] (s4) --++(-1.6,0) --++(0,3.7)-- (t1);
%
%
%
\node [squarenode] (dum0) [right=3.8cm of tt1]{$\omega {\cdot} 2$};
\node [squarenode] (dum1) [above=1cm of dum0]{$\omega {+} 2$};
\node [squarenode] (dum5) [draw=none, above=1cm of dum1]{$\vdots$};
\path[->] (dum0) edge (dum1);
\path[->] (dum1) edge (dum5);
\path[->] (tt1) edge (dum0);
\node [squarenode] (dum2) [draw=none,right=1.5cm of dum0]{$\cdots$};
\node [squarenode] (dum6) [draw=none,right=1.5cm of dum1]{$\cdots$};

\path[->] (dum0) edge (dum2);
\path[->] (dum6) edge (dum1);

\node [squarenode] (dum3) [left=1.2cm of tt1]{$\bot$};
\node [roundnode,double] (dum4) [draw,left=1.1cm of dum3]{$\bot$};
\path[->] (dum3) edge (dum4);
\path[->] (dum3) edge (tt1);

\node [squarenode] (dum5) [draw=none, above=1cm of m4]{$\vdots$};
\path[->] (dum5) edge (m4);

\end{tikzpicture}
\caption{A finitely branching reachability game where the states of player~$\pz$
				are drawn as squares and the random states as circles. Player~$\po$ is passive in this game.
The states with double borders form the target set~$\reachset$; those states have self-loops which are not drawn in the figure.
For each random state, the distribution over the successors is uniform.
Each state is labeled with an ordinal, which indicates the \emph{index} of the
state.
In particular, the example shows that transfinite indices are needed.
}
\label{fig:exampleLimitOrdinal}
\end{figure*}
Each state has value~$1$ in this game, except those labeled with~$0$.
However, only the states labeled with~$\undef$ are almost-surely winning for player~$\pz$.
To see this, consider a player~$\pz$ state labeled with~$1$.
In order to reach~$\reachset$, player~$\pz$ eventually needs to take a transition to a $0$-labeled state, which is not almost-surely winning.
This means that the $1$-labeled states are not almost-surely winning either.
Hence, player~$\pz$ cannot risk entering them if the player wants to win almost surely.
Continuing this style of reasoning, we infer that the $2$-labeled states are not almost-surely winning, and so on.
This implies that the $\omega$-labeled states are not almost-surely winning, and so on.
The only almost-surely winning player~$\pz$ state is the $\undef$-labeled state at the bottom of the figure, and the only winning strategy is to take the direct transition to the target in the bottom-left corner.
\end{example}

\newcommand{\oldgame}{\hat\game}%
\newcommand{\oldstates}{\hat\states}%
\newcommand{\oldstate}{\hat\state}%
\newcommand{\oldformula}{\hat\formula}%
\begin{proof}[Proof of Theorem~\ref{thm:amost-sure-strong-det}]
The first step of the proof is to transform the game and the objective so that the objective can in some respects be treated like a tail objective.
Let $\oldgame$ be a stochastic game with countable state space~$\oldstates$ and objective~$\oldformula$.
We convert the game graph to a forest by encoding the history in the states.
Formally we proceed as follows.
The state space, $\states$, of the new game, $\game$, consists of the partial plays in~$\oldgame$, i.e., $\states \subseteq \oldstates^* \oldstates$.
Observe that $\states$ is countable.
For any $\xsymbol \in \{\zsymbol,\osymbol,\rsymbol\}$ we define $\states_{\xsymbol} := \{w \oldstate \in \states \mid \oldstate \in \oldstates_{\xsymbol}\}$.
A transition is a transition of~$\game$ iff it is of the form $w \oldstate \transition w \oldstate \oldstate'$ where $w \oldstate \in \states$ and $\oldstate \transition \oldstate'$ is a transition in~$\oldgame$.
The probabilities in~$\game$ are defined in the obvious way.
For $\oldstate \in \oldstates$ we define an objective~$\formula_{\oldstate}$
so that a play in~$\game$ starting from the singleton $\oldstate\in\states$ 
satisfies~$\formula_{\oldstate}$ iff the corresponding play from
$\oldstate\in\oldstates$ in~$\oldgame$ satisfies~$\oldformula$.
Since strategies in~$\game$ (for singleton initial states in~$\oldstates$) carry over to strategies in~$\oldgame$, it suffices to prove our determinacy result for~$\game$.

Let us inductively extend the definition of~$\formula_s$ from
$\state=\oldstate\in\oldstates$
to arbitrary $s \in \states$.
For any transition $s \transition s'$ in~$\game$, define $\formula_{s'} := \{x \in s' \states^\omega \mid s x \in \formula_s\}$.
This is well-defined as the transition graph of~$\game$ is a forest.
For any $s \in \states$, the event~$\formula_{s}$ is also measurable.
By this construction 
we obtain the following property:
If a play~$y$ in~$\game$ visits states $s, s' \in \states$ then the suffix of~$y$ starting from~$s$ satisfies~$\formula_s$ iff the suffix of~$y$ starting from~$s'$ satisfies~$\formula_{s'}$.
This property is weaker than the tail property (which would stipulate that all $\formula_s$ are equivalent), but it suffices for our purposes.

In the remainder of the proof, when $\game'$ is (a subgame of)~$\game$, we write $\probm_{\game',s,\zstrat,\ostrat}(\formula)$ for $\probm_{\game',s,\zstrat,\ostrat}(\formula_s)$ to avoid clutter.
Similarly, when we write $\valueof{\game'}{s}$ we mean the value with respect to~$\formula_s$.

In order to characterize the winning sets of the players, we construct a transfinite sequence of subgames $\game_\alpha$ of~$\game$, 
where $\alpha \in \ord$ is an ordinal number, by stepwise removing
certain states that are losing for player~$\pz$, along with their incoming
transitions. Thus some subgames $\game_\alpha$ may contain states without any
outgoing transitions (i.e., dead ends). Such dead ends are always
considered as losing for player~$\pz$. (Formally, one might add a self-loop to
such states and remove from the objective all plays that reach these states.)

Let $\states_\alpha$ denote the state space of the subgame~$\game_\alpha$.
We start with $\game_0 := \game$.
Given~$\game_\alpha$, denote by~$D_\alpha$ the set of states $s \in \states_\alpha$ with $\valueof{\game_\alpha}{s} < 1$.
For any $\alpha \in \ord\setminus\{0\}$ we define $\states_\alpha := \states \setminus \bigcup_{\gamma < \alpha} D_\gamma$.

Since the sequence of sets $\states_\alpha$
is non-increasing and $\states_0 = \states$ is countable, it follows that this
sequence of games $\game_\alpha$ converges (i.e., is ultimately constant)
at some ordinal~$\beta$ where $\beta \le \omega_1$ (the first uncountable ordinal).
That is, we have $\game_\beta = \game_{\beta +1}$.
Note in particular that $\game_\beta$ does not contain any dead
ends. (However, its state space $\states_\beta$ might be empty. In this case
it is considered to be losing for player~$\pz$.)

We define the \emph{index}, $I(\state)$, of a state $\state$ as the 
smallest ordinal~$\alpha$ with $s \in D_\alpha$, and as $\undef$ if such an ordinal does not exist.
For all states~$s \in \states$ we have:
\[
I(s) = \undef \ \Leftrightarrow \ s \in \states_\beta \ \Leftrightarrow \ \valueof{\game_\beta}{s} = 1
\]
We show that states~$s$ with $I(s) \in \ord$ are in ${\owinset{\formula}{<1}}_{\!\!\game}$, and states~$s$ with $I(s) = \undef$ are in ${\zwinset{\formula}{=1}}_{\!\!\game}$.

\smallskip{\noindent\bf Strategy~$\hat\ostrat_s$:}
For each $s \in \states$ with $I(s) \in \ord$ we construct a player~$\po$ strategy~$\hat\ostrat_s$ such that $\probm_{\game,s,\zstrat,\hat\ostrat_s}(\formula) < 1$ holds for all player~$\pz$ strategies~$\zstrat$. 
The strategy~$\hat\ostrat_s$ is defined inductively over the index~$I(s)$.

Let $s \in \states$ with $I(s) = \alpha \in \ord$.
In game~$\game_\alpha$ we have $\valueof{\game_\alpha}{s} < 1$.
So by weak determinacy (Theorem~\ref{thm:weak_borel_determinacy}) there is a strategy~$\hat\ostrat_s$ with $\probm_{\game_\alpha,s,\zstrat,\hat\ostrat_s}(\formula) < 1$ for all~$\zstrat$.
(For example, one may take a $(1-\valueof{\game_\alpha}{s})/2$-optimal player~$\po$ strategy).
We extend $\hat\ostrat_s$ to a strategy in~$\game$ as follows.
Whenever the play enters a state $s' \notin \states_\alpha$ (hence $I(s') < \alpha$) then $\hat\ostrat_s$ switches to the previously defined strategy $\hat\ostrat_{s'}$.
(One could show that only player~$\pz$ can take a transition leaving~$\states_\alpha$, although this is not needed at the moment.)

We show by transfinite induction on the index that
$\probm_{\game,s,\zstrat,\hat\ostrat_s}(\formula) < 1$ holds for all
player~$\pz$ strategies~$\zstrat$ and for all states $s \in \states$ with
$I(s) \in \ord$.

For the induction hypothesis, let $\alpha$ be an ordinal for which this holds for all states~$s$ with $I(s) < \alpha$.
For the inductive step, let $s \in \states$ be a state with $I(s) = \alpha$, and let $\zstrat$ be an arbitrary player~$\pz$ strategy in~$\game$.

Suppose that the play from~$s$ under the strategies $\zstrat,\hat{\ostrat}_s$ always remains in~$\states_\alpha$, i.e., the probability of ever leaving~$\states_\alpha$ under $\zstrat,\hat{\ostrat}_s$ is zero.
Then any play in~$\game$ under these strategies coincides with a play in~$\game_\alpha$, so we have $\probm_{\game,s,\zstrat,\hat\ostrat_s}(\formula) = \probm_{\game_\alpha,s,\zstrat,\hat\ostrat_s}(\formula) < 1$, as desired.
Now suppose otherwise, i.e., the play from~$s$ under $\zstrat,\hat{\ostrat}_s$, 
with positive probability, enters a state $s' \notin \states_\alpha$, hence $I(s') < \alpha$.
By the induction hypothesis we have $\probm_{\game,s',\zstrat',\hat\ostrat_{s'}}(\formula) < 1$ for any~$\zstrat'$.
Since the probability of entering~$s'$ is positive, we conclude $\probm_{\game,s,\zstrat,\hat\ostrat_s}(\formula) < 1$, as desired.

\smallskip{\noindent\bf Strategy~$\hat\zstrat$:}
For each $s \in \states$ with $I(s) = \undef$ (and thus $\state \in
\states_\beta$) we construct a player~$\pz$ strategy~$\hat\zstrat$ such that $\probm_{\game,s,\hat\zstrat,\ostrat}(\formula) = 1$ holds for all player~$\po$ strategies~$\ostrat$.
We first observe that if $s_1 \transition s_2$ is a transition in~$\game$ with $s_1 \in \ostates \cup \rstates$ and $I(s_2) \ne \undef$ then $I(s_1) \ne \undef$.
Indeed, let $I(s_2) = \alpha \in \ord$, thus $\valueof{\game_\alpha}{s_2} < 1$;
if $s_1 \in \states_\alpha$ 
then $\valueof{\game_\alpha}{s_1} < 1$ and thus $I(s_1) = \alpha$;
if $s_1 \notin \states_\alpha$ then $I(s_1) < \alpha$.
It follows that only player~$\pz$ could ever leave the state space~$\states_\beta$,
but our player~$\pz$ strategy~$\hat\zstrat$ will ensure that the play remains in~$\states_\beta$ forever.
Recall that 
$\game_\beta$ does not contain any dead ends and that
$\valueof{\game_\beta}{s} = 1$ for all $s \in \states_\beta$.
For all $s \in \states_\beta$, by weak determinacy (Theorem~\ref{thm:weak_borel_determinacy}) we fix a strategy~$\zstrat_s$ with $\probm_{\game_\beta,s,\zstrat_s,\ostrat}(\formula) \ge 2/3$ for all~$\ostrat$.

Fix an arbitrary state $s_0 \in \states_\beta$ as the initial state.
For a player~$\pz$ strategy~$\zstrat$, define mappings $X^\zstrat_1, X^\zstrat_2, \ldots : s_0 \states^\omega \to [0,1]$ using conditional probabilities:
\[
X^\zstrat_i(w) := \inf_{\ostrat \in \ostratset_{\game_\beta}} \probm_{\game_\beta,s_0,\zstrat,\ostrat}(\formula \mid E_i(w))\,,
\]
where $E_i(w)$ denotes the event containing the plays that start with the length-$i$ prefix of~$w \in s_0 \states^\omega$.
Thanks to our ``forest'' construction at the beginning of the proof, $X^\zstrat_i(w)$
depends, in fact, only on the $i$-th state visited by~$w$.

For some illustration, a small value of~$X^\zstrat_i(w)$ means that considering the length-$i$ prefix of~$w$, player~$\po$ has a strategy that makes~$\formula$ unlikely at time~$i$.
Similarly, a large value of~$X^\zstrat_i(w)$ means that at time~$i$ (when the length-$i$ prefix has been ``uncovered'') the probability of~$\formula$ using~$\zstrat$ is large, regardless of the player~$\po$ strategy.

In the following we view~$X^\zstrat_i$ as a random variable (taking on a random value depending on a random play).

We define our almost-surely winning player~$\pz$ strategy~$\hat\zstrat$ as the limit of inductively defined strategies $\hat\zstrat_0, \hat\zstrat_1, \ldots$.
Let $\hat\zstrat_0 := \zstrat_{s_0}$.
Using the definition of~$\zstrat_{s_0}$ we get $X^{\hat\zstrat_0}_1 \ge 2/3$.
For any $k \in \N$, define $\hat\zstrat_{k+1}$ as follows.
Strategy~$\hat\zstrat_{k+1}$ plays~$\hat\zstrat_k$ as long as $X^{\hat\zstrat_k}_i \ge 1/3$.
This could be forever.
Otherwise, let $i$ denote the smallest~$i$ with $X^{\hat\zstrat_k}_i < 1/3$, and let $s$ be the $i$-th state of the play.
At that time, $\hat\zstrat_{k+1}$ switches to strategy~$\zstrat_{s}$, implying $X^{\hat\zstrat_{k+1}}_i \ge 2/3$.
This switch of strategy is referred to as a ``reset'' in~\cite{GimbertHornSoda10}, where the concept is used similarly.
For any~$k$, strategy~$\hat\zstrat_k$ performs at most~$k$ such resets.
Define $\hat\zstrat$ as the limit of the $\hat\zstrat_k$, i.e., the number of resets performed by~$\hat\zstrat$ is unbounded.

In order to show that $\hat\zstrat$ is almost surely winning, we first argue that $\hat\zstrat$ almost surely performs only a finite number of resets.
Suppose $w \in S^\omega$ and $k,i$ are such that a $k$-th reset happens after visiting the $i$-th state in~$w$.
As argued above, we have $X^{\hat\zstrat_k}_i(w) \ge 2/3$.
Towards a contradiction assume that player~$\po$ has a strategy~$\ostrat_1$ to cause yet another reset with probability $p_1>1/2$, i.e.,
\[
p_1 := \probm_{\game_\beta,s_0,\hat\zstrat_k,\ostrat_1}(R \mid E_i(w)) > 1/2\,,
\]
where $R$ denotes the event of another reset after time~$i$.
If another reset occurs, say at time~$j$, then $X^{\hat\zstrat_k}_j(w) < 1/3$, and then player~$\po$ can switch to a strategy~$\ostrat_2$ to force $\probm_{\game_\beta,s_0,\hat\zstrat_k,\ostrat_2}(\formula \mid E_j(w)) \le 1/3$.
Hence:
\[
p_2 := \probm_{\game_\beta,s_0,\hat\zstrat_k,\ostrat_2}(\formula \mid R \land E_i(w)) \le 1/3
\]
Let $\ostrat_{1,2}$ denote the player~$\po$ strategy combining $\ostrat_1$ and~$\ostrat_2$.
Then it follows:
\begin{align*}
\probm_{\game_\beta,s_0,\hat\zstrat_k,\ostrat_{1,2}}(\formula \land R \mid E_i(w)) \ & = \ p_1 \cdot p_2 \qquad \text{and} \\
\probm_{\game_\beta,s_0,\hat\zstrat_k,\ostrat_{1,2}}(\formula \land \neg R \mid E_i(w)) \ & \le \ \probm_{\game_\beta,s_0,\hat\zstrat_k,\ostrat_{1,2}}(\neg R \mid E_i(w)) \\ 
& = \ 1-p_1
\end{align*}
Hence we have:
\begin{align*}
\probm_{\game_\beta,s_0,\hat\zstrat_k,\ostrat_{1,2}}(\formula \mid E_i(w)) 
\ & \le \ (p_1 \cdot p_2) + (1-p_1) \ \le \ 1 - \frac23 p_1 \\
  &  < \ 1 - \frac23 \cdot \frac12 \ = \ \frac23\,,
\end{align*}
contradicting $X^{\hat\zstrat_k}_i(w) \ge 2/3$.
So at time~$i$, the probability of another reset is bounded by~$1/2$.
Since this holds for every reset time~$i$, we conclude that almost surely there will be only finitely many resets under~$\hat\zstrat$, regardless of~$\ostrat$.

Now we can show that $\probm_{\game_\beta,s_0,\hat\zstrat,\ostrat}(\formula) = 1$ holds for all~$\ostrat$.
Fix~$\ostrat$ arbitrarily.
For $k \in \N$ define~$Q_k$ as the event that exactly $k$ resets occur.
Let us write $\probm_k = \probm_{\game_\beta,s_0,\hat\zstrat_k,\ostrat}$ to avoid clutter.
By L\'evy's zero-one law (see, e.g., \cite[Theorem~14.2]{Williams:Martingales}), for any~$k$, we have $\probm_k$-almost surely that either
\[
 (\formula \lor \neg Q_k) \land \lim_{i \to \infty} \probm_k(\formula \lor \neg Q_k \mid E_i(w)) = 1
\]
or
\[
 (\neg\formula \land Q_k) \land \lim_{i \to \infty} \probm_k(\formula \lor \neg Q_k \mid E_i(w)) = 0
\]
holds.
Let $w$ be a play that satisfies the second option.
In particular, $w \in Q_k$, so there exists $i_0 \in \N$ with $X^{\hat\zstrat_k}_i(w) \ge 1/3$ for all $i \ge i_0$.
It follows that $\probm_k(\formula \mid E_i(w)) \ge 1/3$ holds for all $i \ge i_0$.
But that contradicts the fact that $\lim_{i \to \infty} \probm_k(\formula \lor \neg Q_k \mid E_i(w)) = 0$.
So plays satisfying the second option do not actually exist.

Hence we conclude $\probm_k(\formula \lor \neg Q_k) = 1$, thus $\probm_k(\neg\formula \land Q_k) = 0$.
Since the strategies~$\hat\zstrat$ and~$\hat\zstrat_k$ agree on all finite prefixes of all plays in~$Q_k$,
the probability measures $\probm_{\game_\beta,s_0,\hat\zstrat,\ostrat}$ and~$\probm_k$ agree on all subevents of~$Q_k$.
It follows $\probm_{\game_\beta,s_0,\hat\zstrat,\ostrat}(\neg\formula \land Q_k) = 0$.
We have shown previously that the number of resets is almost surely finite, i.e., $\probm_{\game_\beta,s_0,\hat\zstrat,\ostrat}(\bigvee_{k \in \N} Q_k) = 1$.
Hence we have:
\begin{align*}
\probm_{\game_\beta,s_0,\hat\zstrat,\ostrat}(\neg\formula) \
& = \ \probm_{\game_\beta,s_0,\hat\zstrat,\ostrat}\Big(\neg\formula \land \bigvee_{k \in \N} Q_k\Big) \\
& \le \ \sum_{k \in \N} \probm_{\game_\beta,s_0,\hat\zstrat,\ostrat}(\neg\formula \land Q_k) \\
& = \ 0
\end{align*}
Thus, $\probm_{\game_\beta,s_0,\hat\zstrat,\ostrat}(\formula) = 1$.
Since $\hat\zstrat$ is defined on~$\game_\beta$, this strategy never leaves~$\states_\beta$.
Since only player~$\pz$ might have transitions that leave~$\states_\beta$, we conclude $\probm_{\game,s_0,\hat\zstrat,\ostrat}(\formula) = 1$.
\end{proof}

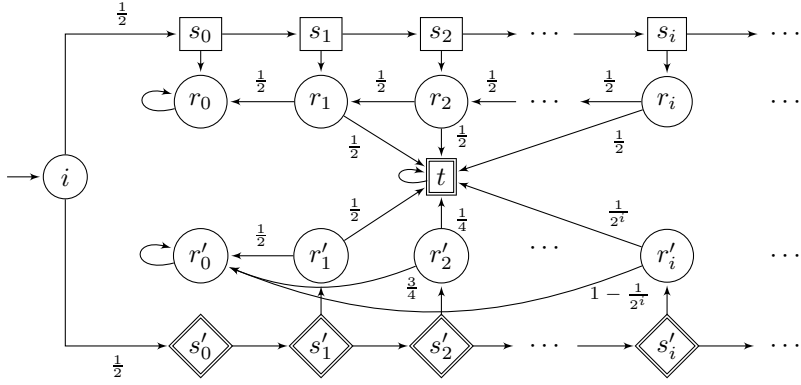
\begin{figure*}[t]
\centering
\begin{tikzpicture}[>=latex',shorten >=1pt,node distance=1.9cm,on grid,auto,
roundnode/.style={circle, draw,minimum size=1.5mm},
squarenode/.style={rectangle, draw,minimum size=2mm},
diamonddnode/.style={diamond,draw, minimum size=8mm, inner sep=-10pt}]

\node [squarenode] (s0) at(0,0) [draw]{$s_0$};
\node [squarenode] (s1) [draw,right=1.6cm of s0]{$s_1$};
\node [squarenode] (s2) [draw,right=1.6cm of s1]{$s_2$};
\node [squarenode] (s3) [draw=none,right=1.4cm of s2]{$\cdots$}; 
\node [squarenode] (s4) [draw,right=1.6cm of s3]{$s_i$}; 
\node [squarenode] (s5) [draw=none,right=1.6cm of s4]{$\cdots$}; 
 
\node[roundnode] (r0)  [below=.9cm of s0] {$r_0$};
\node[roundnode] (r1)  [below=.9cm of s1] {$r_1$};
\node[roundnode] (r2)  [below=.9cm of s2] {$r_2$};
\node[roundnode] (r3)  [draw=none,below=.9 of s3] {$\cdots$};
\node[roundnode] (r4)  [below=.9cm of s4] {$r_i$};
\node[roundnode] (r5)  [draw=none,below=.9cm of s5] {$\cdots$};

\node [squarenode,double] (t)  [below=1cmof r2] {$ t $};

\node [diamonddnode,double] (ss0) [draw,below=3.25cm of r0] [draw]{$s'_0$};
\node [diamonddnode,double] (ss1) [draw,right=1.6cm of ss0]{$s'_1$};
\node [diamonddnode,double] (ss2) [draw,right=1.6cm of ss1]{$s'_2$};
\node [diamonddnode,double] (ss3) [draw=none,right=1.4cm of ss2]{$\cdots$}; 
\node [diamonddnode,double] (ss4) [draw,right=1.6cm of ss3]{$s'_i$}; 
\node [diamonddnode,double] (ss5) [draw=none,right=1.6cm of ss4]{$\cdots$}; 
 
\node[roundnode,inner sep=0.9mm] (rr0)  [below=-1.2cm of ss0] {$r'_0$};
\node[roundnode,inner sep=0.9mm] (rr1)  [below=-1.2cm of ss1] {$r'_1$};
\node[roundnode,inner sep=0.9mm] (rr2)  [below=-1.2cm of ss2] {$r'_2$};
\node[roundnode,inner sep=0.9mm] (rr3)  [draw=none,below=-1.3 of ss3] {$\cdots$};
\node[roundnode,inner sep=0.9mm] (rr4)  [below=-1.2cm of ss4] {$r'_i$};
\node[roundnode,inner sep=0.9mm] (rr5)  [draw=none,below=-1.2cm of ss5] {$\cdots$};

\path[->] (s0) edge (s1);
\path[->] (s1) edge (s2);
\path[->] (s2) edge (s3);
\path[->] (s3) edge (s4);
\path[->] (s4) edge (s5);

\path[->] (s0) edge (r0);
\path[->] (s1) edge (r1);
\path[->] (s2) edge (r2);
\path[->] (s4) edge (r4);

\path[->] (r1) edge node [very near start,below=.05cm] {\scriptsize{$\frac{1}{2}$}} (t);
\path[->] (r2) edge node [near start,right=0.05cm] {\scriptsize{$\frac{1}{2}$}} (t);
\path[->] (r4) edge node [very near start,below=.05cm] {\scriptsize{$\frac{1}{2}$}} (t);

\path[->] (r1) edge node[above] {\scriptsize{$\frac{1}{2}$}} (r0);   
\path[->] (r2) edge node [above] {\scriptsize{$\frac{1}{2}$}} (r1);  
\path[->] (r3) edge node [above] {\scriptsize{$\frac{1}{2}$}} (r2);
\path[->] (r4) edge node [above] {\scriptsize{$\frac{1}{2}$}} (r3);

\path[->] (t) edge [loop left]   ();	
\path[->] (r0) edge [loop left]   ();	
\path[->] (rr0) edge [loop left]   ();	

\path[->] (ss0) edge (ss1);
\path[->] (ss1) edge (ss2);
\path[->] (ss2) edge (ss3);
\path[->] (ss3) edge (ss4);
\path[->] (ss4) edge (ss5);

\path[->] (ss1) edge (rr1);
\path[->] (ss2) edge (rr2);
\path[->] (ss4) edge (rr4);

\path[->] (rr1) edge node [very near start,above=.05cm] {\scriptsize{$\frac{1}{2}$}} (t);
\path[->] (rr2) edge node [near start,right=0.05cm] {\scriptsize{$\frac{1}{4}$}} (t);
\path[->] (rr4) edge node [very near start,above=.05cm] {\scriptsize{$\frac{1}{2^i}$}} (t);

\path[->] (rr1) edge node[above] {\scriptsize{$\frac{1}{2}$}} (rr0);   
\path[->] (rr2) edge [bend left=22] node [pos=0.02,below] {\scriptsize{$\frac{3}{4}$}} (rr0);  
\path[->] (rr4) edge [bend left=22] node [pos=0.05,below] {\scriptsize{$1-\frac{1}{2^{i}}$}} (rr0);  

\node[roundnode,initial,initial text={}] (i)  [left=5cm of t] {$i$};
\draw [->] (i) --++(0,1.9)-- (s0) node[draw=none,fill=none,font=\scriptsize,midway,above] {\scriptsize{$\frac{1}{2}$}};
\draw [->] (i) --++(0,-2.25)-- (ss0) node[draw=none,fill=none,font=\scriptsize,midway,below] {\scriptsize{$\frac{1}{2}$}};

\end{tikzpicture} 
\caption{A finitely branching game where the states  of players $\pz$ and~$\po$
				are drawn as squares and diamonds, respectively;   random states~$s\in \rstates$
         are drawn as circles.
                                The states $s_i'$ and state $t$ (double
                                borders) are B\"uchi states, all other
                                states are not.
				The value of the initial state~$i$ is
                                $\frac{1}{2}$, for the B\"uchi objective~$\formula$. 
				However, $i\not \in \zwinset{\formula}{\geq \frac{1}{2}} \uplus 
\owinset{\formula}{\not > \frac{1}{2}}$, meaning that neither player has a winning strategy, neither for the objective $\quantobj{\formula}{\mathord{\ge} 1/2}$ nor for $\quantobj{\formula}{\mathord{>} 1/2}$.}
\label{fig:BuchiNOStrongDet}
\end{figure*}

\subsection{Reachability and Safety}

It was shown in~\cite{BBKO:IC2011} and~\cite{Krcal:Thesis:2009} 
(and also follows as a corollary from \cite{Brozek:TCS2013})
that finitely branching games with reachability objectives with any threshold
$\constraint\const$ with $\const \in [0,1]$ are strongly determined.
In contrast, strong determinacy 
does not hold for infinitely branching reachability games 
with thresholds $\constraint\const$ with $\const \in (0,1)$; cf.\ Figure 1 in \cite{BBKO:IC2011}.
However, by Theorem~\ref{thm:amost-sure-strong-det},
strong determinacy does hold for almost-sure reachability and safety objectives in infinitely branching games.
By duality, this also holds for reachability and safety objectives with threshold~$\mathord{>}0$.
(For almost-sure safety (resp.\ $>0$ reachability), 
this could also be shown by a reduction to
non-stochastic 2-player reachability games \cite{Zielonka:1998}.)

\subsection{B\"uchi and co-B\"uchi}

Let $\formula$ be the B\"uchi objective (the co-B\"uchi objective is dual).
Again, Theorem~\ref{thm:amost-sure-strong-det} applies to almost-sure and positive-probability B\"uchi and co-B\"uchi objectives, so those games are strongly determined, even infinitely branching ones.

However, this does not hold for thresholds $c \in (0,1)$, not even for finitely branching games:

\newcommand{\thmBuchinostrongdet}{
Threshold (co-)B\"uchi objectives $\quantobj{\formula}{\constraint\const}$
with thresholds $\const \in (0,1)$
are not strongly determined, even for
finitely branching games.
}
\begin{theorem}\label{thm:Buchi-no-strong-det}
\thmBuchinostrongdet
\end{theorem}
A fortiori, threshold parity objectives are not strongly determined, not even for finitely branching games.
We prove Theorem~\ref{thm:Buchi-no-strong-det} using the finitely branching
game in Figure~\ref{fig:BuchiNOStrongDet}.
It is inspired by an infinitely branching example in~\cite{BBKO:IC2011}, where it was shown that threshold reachability objectives in infinitely branching games are not strongly determined.

\begin{proof}[Proof sketch of Theorem~\ref{thm:Buchi-no-strong-det}]
The game in Figure~\ref{fig:BuchiNOStrongDet} is finitely branching,
and we consider the B\"uchi objective.
The infinite choice for player~$\po$ in the example of~\cite{BBKO:IC2011} is simulated with an infinite chain~$s'_0 s'_1 s'_2 \cdots$ of 
B\"uchi states in our example. All  states~$s'_0 s'_1 s'_2 \cdots$ are finitely branching and belong to player~$\po$.
The crucial property is that player~$\po$ can stay in the states~$s'_i$ for
arbitrarily long (thus making the probability of reaching the state $t$ arbitrarily
small) but not forever. Since the states~$s'_i$ are B\"uchi states, plays that
stay in them forever satisfy the B\"uchi objective surely, something that
player~$\po$ needs to avoid. 
So a player~$\po$ strategy must choose a transition $s'_i\transition{}r'_i$ for some~$i\in \nat$, 
resulting in a faithful simulation of infinite branching from $s_0'$ to some state
$r_i'$, just like in the reachability game in~\cite{BBKO:IC2011}.

From the fact that~$\valueof{\game}{r_i}=1-2^{-i}$ and $\valueof{\game}{r'_i}=2^{-i}$, we deduce the following properties of this game:
\begin{itemize}
	\item $\valueof{\game}{s_0}=1$, but there exists no optimal strategy starting in~$s_0$. 
	The value is witnessed by a family of $\epsilon$-optimal strategies~$\sigma_i$: 
	traversing the ladder~$s_0 s_1\cdots s_i$ and  choosing~$s_i\transition{r_i}$.
	\item $\valueof{\game}{s'_0}=0$, but there exists no optimal minimizing strategy starting in~$s'_0$; however,
	in analogy with~$s_i$, there are~$\epsilon$-optimal strategies. 
	\item  $\valueof{\game}{i}=\frac{1}{2}$. We argue below that neither player has an optimal strategy starting in~$i$.
It follows that $i\not \in \zwinset{\formula}{\geq \frac{1}{2}} \uplus 
\owinset{\formula}{\not >\frac{1}{2}}$ for the B\"uchi condition~$\varphi$.
So neither player has a winning strategy, neither for $\quantobj{\formula}{\mathord{\ge} 1/2}$ nor for $\quantobj{\formula}{\mathord{>} 1/2}$.
Indeed, consider any player~$\pz$ strategy~$\zstrat$.
Following~$\sigma$,
once the game is in~$s_0$, B\"uchi states cannot be visited with probability more than~$\frac{1}{2}\cdot (1-\epsilon)$
for some fixed $\epsilon>0$ and all strategies~$\pi$. 
Player~$\po$ has an $\frac{\epsilon}{2}$-optimal strategy~$\ostrat$ starting in~$s'_0$.
Then we have:
$$
\probm_{\game,i,\zstrat,\ostrat}(\formula) \le \frac{1}{2}\cdot (1-\epsilon) +
\frac{1}{2}\cdot\frac{\epsilon}{2} < \frac{1}{2}\,,$$
so $\sigma$ is not optimal.
One can argue symmetrically that player~$\po$ does not have an optimal strategy either.
\end{itemize}
In the example in Figure~\ref{fig:BuchiNOStrongDet}, the game branches from
state $i$ to $s_0$ and $s_0'$ with probability $1/2$ respectively.
However, the above argument can be adapted to work for probabilities $\const$ and
$1-\const$ for every constant $\const \in (0,1)$. 
\end{proof}

\section{Memory Requirements}\label{sec:memory}
In this section we study how much memory is needed to win 
objectives $\quantobj{\formula}{\constraint\const}$, depending on
$\formula$ and on the 
constraint $\constraint\const$.

We say that an objective $\quantobj{\formula}{\constraint\const}$ is
\emph{strongly MD-determined} iff for every state $\state$ either
\begin{itemize}
\item  
there exists 
an MD-strategy~$\zstrat$ such that, for all $\ostrat \in \ostratset$,  
we have $\probm_{\game,\state,\zstrat,\ostrat}(\formula) \constraint \const$,
or
\item there exists 
an MD-strategy~$\ostrat$ such that, for all $\zstrat \in \zstratset$,  
we have $\probm_{\game,\state,\zstrat,\ostrat}(\formula) \nconstraint \const$.
\end{itemize}
If a game is strongly MD-determined then it is also strongly determined,
but not vice-versa.
Strong FR-determinacy is defined analogously.

\subsection{Reachability and Safety Objectives}\label{sec:reachability}

Let $\reachset \subseteq \states$ and
$\quantobj{\reach{\reachset}}{\constraint\const}$
be a threshold reachability objective.
(Safety objectives are dual to reachability.) 

Let us briefly discuss infinitely branching reachability games.
If $\const \in (0,1)$ then strong determinacy does not hold; cf.\ Figure 1 in \cite{BBKO:IC2011}.
Objectives $\quantobj{\reach{\reachset}}{\ge 1}$ are strongly determined
(Theorem~\ref{thm:amost-sure-strong-det}), but not strongly
FR-determined, because player~$\po$ needs infinite memory (even if player~$\pz$
is passive)~\cite{Kucherabook11}.
Objectives $\quantobj{\reach{\reachset}}{> 0}$ correspond to non-stochastic
2-player reachability games, which are strongly MD-determined \cite{Zielonka:1998}.

In the rest of this subsection we consider finitely branching reachability games. 
It is shown in~\cite{BBKO:IC2011,Krcal:Thesis:2009} that finitely branching reachability games are
strongly determined, but the winning $\pz$ strategy constructed therein uses infinite memory.
Indeed, Ku\v{c}era~\cite{Kucherabook11} showed that infinite memory is necessary in general:
\begin{theorem}[follows from Proposition~5.7.b in~\cite{Kucherabook11}]\label{thm:reach_not_fr_det} 
Finitely branching reachability games with 
$\quantobj{\reach{\reachset}}{\ge\const}$ objectives are not
strongly FR-determined for $\const \in (0,1)$.
\end{theorem}

The example from~\cite{Kucherabook11} that proves Theorem~\ref{thm:reach_not_fr_det} has the following properties:
\begin{enumerate}
\item[(1)] player~$\pz$ has \emph{value-decreasing} (see below) transitions;
\item[(2)] player~$\po$ has \emph{value-increasing} (see below) transitions;
\item[(3)] threshold $c \ne 0$ and $c \ne 1$;
\item[(4)] nonstrict inequality: ${\ge} c$.
\end{enumerate}
Given a game~$\game$, we call a transition $\state \transition \state'$ \emph{value-decreasing (resp., value-increasing)} if $\valueof{\game}{\state} > \valueof{\game}{\state'}$ (resp., $\valueof{\game}{\state} < \valueof{\game}{\state'}$).
If player~$\pz$ (resp., player~$\po$) controls a transition $\state \transition \state'$, i.e., $\state \in \zstates$ (resp., $\state \in \ostates$), then the transition cannot be value-increasing (resp., value-decreasing).
We write $\rvi(\game)$ for the game obtained from~$\game$ by removing the
value-increasing transitions controlled by player~$\po$. 
Note that this operation does not create any dead ends in finitely branching games, 
because at least one transition to a successor state with the same value will
always remain for such games.

We show that a reachability game is strongly MD-determined if any of the properties listed above is not satisfied:
\begin{theorem}\label{thm:reachability}
Finitely branching games $\game$ with reachability objectives
$\quantobj{\reach{\reachset}}{\constraint\const}$
are strongly MD-determined, provided that at least one of the following
conditions holds.
\begin{enumerate}
\item[(1)] player~$\pz$ does not have value-decreasing transitions, or
\item[(2)] player~$\po$ does not have value-increasing transitions, or
\item[(3)] almost-sure objective: $\mathord{\constraint} = \mathord{\ge}$ and $\const=1$, or
\item[(4)] strict inequality: $\mathord{\constraint} = \mathord{>}$.
\end{enumerate}
\end{theorem}
\begin{remark}
\rm
Condition (1) or~(2) of Theorem~\ref{thm:reachability} is trivially satisfied
if the corresponding player is passive, i.e., in MDPs.
It was already known that MD strategies are sufficient for safety and reachability
objectives in countable finitely branching MDPs (\cite{Puterman:book}, Section
7.2.7).
Theorem~\ref{thm:reachability} generalizes this result.
\end{remark}

\begin{remark}\label{rem:inf-branching-inf-memory}
\rm
Theorem~\ref{thm:reachability} does not carry over to
stochastic reachability games with an arbitrary number of players, not even
if the game graph is finite.
Instead multiplayer games can require infinite memory to win.
Proposition 4.13 in \cite{Ummels-Wojtczak:LMS2011} 
constructs an
11-player finite-state stochastic reachability game with a pure
subgame-perfect Nash equilibrium where the first player wins almost surely
by using infinite memory. However, there is no finite-state Nash equilibrium
(i.e., an equilibrium where all players are limited to finite memory)
where the first player wins with positive probability.
That is, the first player cannot win with only finite memory, not even if the other players
are restricted to finite memory.
\end{remark}
The rest of the subsection focuses on the proof of Theorem~\ref{thm:reachability}.
We will need the following result from~\cite{BBKO:IC2011}:

\begin{lemma}\label{lem:optmin}{\bf (Theorem~3.1 in \cite{BBKO:IC2011})}
If $\game$ is a finitely branching reachability game then there is an MD strategy 
$\ostrat \in \ostratset$ that is optimal minimizing in every $\po$ state
(i.e., $\valueof{\game}{\ostrat(\state)} = \valueof{\game}{\state}$). 
\end{lemma}
One challenge in proving Theorem~\ref{thm:reachability} is that an optimal minimizing player~$\po$ MD strategy according to Lemma~\ref{lem:optmin} is not necessarily winning for player~$\po$, even for almost-sure reachability and even if player~$\po$ has a winning strategy.
Indeed, consider the game in Figure~\ref{fig:BuchiNOStrongDet}, and add a new player~$\po$ state $u$
and transitions $u \transition s_0$ and $u \transition t$.
For the reachability objective $\reach{\{t\}}$, we then have $\valueof{\game}{u} = \valueof{\game}{s_0} = \valueof{\game}{t} = 1$, and the player~$\po$ MD strategy $\ostrat$ with $\ostrat(u) = t$ is optimal minimizing.
However, $\po$ is not winning from $u$ w.r.t.\ the almost-sure objective $\quantobj{\reach{\{t\}}}{\ge 1}$.
Instead the winning strategy is $\ostrat'$ with $\ostrat'(u) = s_0$.

By the following lemma (from~\cite{BBKO:IC2011}), player~$\pz$ has for every state an $\epsilon$-optimal strategy that needs to be defined only on a finite horizon:
\begin{lemma}\label{lem:reach-approx}{\bf (Lemma 3.2 in \cite{BBKO:IC2011})}
If $\game$ is a finitely branching game with reachability objective $\reach{\reachset}$ then: 
\[
\begin{array}{l}
\forall\, \state \in \states\ \forall\, \epsilon >0\ \exists\, \zstrat \in \zstratset\ \exists\, n \in \N\
\forall\, \ostrat \in \ostratset\,.\\
\probm_{\game,\state,\zstrat,\ostrat}(\reachn{n}{\reachset}) >
\valueof{\game}{\state} - \epsilon\,,
\end{array}
\]
where $\reachn{n}{\reachset}$ denotes the event of reaching~$\reachset$ within at most $n$~steps.
\end{lemma}

Towards a proof of item~(1) of Theorem~\ref{thm:reachability}, we prove the following lemma:
\begin{lemma} \label{lem:reach-opt-uniform}
Let $\game$ be a finitely branching game with reachability objective $\reach{\reachset}$.
Suppose that player~$\pz$ does not have any value-decreasing transitions.
Then there exists a player~$\pz$ MD strategy~$\hat\zstrat$ that is optimal in all states.
That is, for all states~$s$ and for all player~$\po$ strategies~$\ostrat$ we have $\probm_{\game,s,\hat\zstrat,\ostrat}(\reach{\reachset}) \ge \valueof{\game}{s}$.
\end{lemma}
\begin{proof}
In order to construct the claimed MD strategy~$\hat{\zstrat}$,
we define a sequence of modified games $\game_i$ in which the strategy
of player~$\pz$ is already fixed on a finite subset of the state space.
We will show that the value of any state remains the same in all the $\game_i$, i.e.,
$\valueof{\game_i}{\state} = \valueof{\game}{\state}$ for all~$\state$.
Fix an enumeration $s_1, s_2, \ldots$ that includes every state in~$\states$ infinitely often.
Let $\game_0 := \game$.

Given~$\game_i$ we construct~$\game_{i+1}$ as follows.
We use Lemma~\ref{lem:reach-approx} to get a strategy $\zstrat_i$ and $n_i \in \N$
s.t.\ $\probm_{\game_i,\state_i,\zstrat_i,\ostrat}(\reachn{n_i}{\reachset}) >
\valueof{\game_i}{s_i} - 2^{-i}$.
From the finiteness of $n_i$ and the assumption that $\game$ is
finitely branching, we obtain that
${\it Env}_i := \{\state\,|\, \state_i \transition^{\le {n_i}} \state\}$
is finite. Consider the subgame~$\game_i'$ with finite state space ${\it Env}_i$.
In this subgame there exists an optimal MD strategy $\zstrat_i'$ that maximizes
the reachability probability for every state in ${\it Env}_i$.
In particular, $\zstrat_i'$ achieves the same approximation in~$\game_i'$ as $\zstrat_i$ 
in~$\game_i$, i.e.,
$\probm_{\game_i',\state_i,\zstrat_i',\ostrat}(\reach{\reachset}) >
\valueof{\game_i}{s_i} - 2^{-i}$.
Let ${\it Env}_i'$ be the subset of states $\state$ in 
${\it Env}_i$ with $\valueof{\game_i'}{\state} >0$.
Since ${\it Env}_i'$ is finite, there exist $n_i' \in \N$ and $\lambda > 0$ with $\probm_{\game_i',\state,\zstrat_i',\ostrat}(\reachn{n_i'}{\reachset}) \ge \lambda$ for all $\state \in {\it Env}_i'$ and all $\ostrat \in \ostratset_{\game_i'}$.

We now construct $\game_{i+1}$ by modifying $\game_i$ as follows.
For every player~$\pz$ state $\state \in {\it Env}_i'$ 
we fix the transition according to 
$\zstrat_i'$, i.e., only transition $\state\transition\zstrat_i'(\state)$
remains and all other transitions from $\state$ are deleted.
Since all moves from $\pz$ states in~${\it Env}_i'$ have been fixed according to~$\zstrat_i'$,
the bounds above for~$\game_i'$ and~$\zstrat_i'$ now hold for~$\game_{i+1}$ and any $\zstrat \in \zstratset_{\game_{i+1}}$.
That is, we have $\probm_{\game_{i+1},\state_i,\zstrat,\ostrat}(\reach{\reachset}) > \valueof{\game_i}{s_i} - 2^{-i}$ and $\probm_{\game_{i+1},\state,\zstrat,\ostrat}(\reachn{n_i'}{\reachset}) \ge \lambda$ for all $s \in {\it Env}_i'$ and all $\zstrat \in \zstratset_{\game_{i+1}}$ and all $\ostrat \in \ostratset_{\game_{i+1}}$.

Now we show that the values of all states $\state$ in 
$\game_{i+1}$ are still the same as in $\game_i$.
Since our games are weakly determined, it suffices to show that player~$\pz$
has an $\epsilon$-optimal strategy from $\state$ in $\game_{i+1}$ 
for every $\epsilon>0$.
Let $\ostrat$ be an arbitrary $\po$ strategy from $\state$ in $\game_{i+1}$.
Let $\state$ be a state and $\zstrat$ be an $\epsilon/2$-optimal $\pz$ strategy from
$\state$ in $\game_i$. We now define a $\pz$ strategy $\zstrat'$ from
$\state$ in $\game_{i+1}$. If the game does not enter ${\it Env}_i'$
then $\zstrat'$ plays exactly as $\zstrat$ (which is possible since outside
${\it Env}_i'$ no transitions have been removed).
If the game enters ${\it Env}_i'$ then it will reach the target 
from within ${\it Env}_i'$ with probability $\ge \lambda$.
Moreover, if the game stays inside ${\it Env}_i'$ forever then it will almost surely reach
the target, since $(1-\lambda)^\infty=0$. Otherwise, it exits ${\it Env}_i'$
at some state $\state' \notin {\it Env}_i'$ (strictly speaking, at a
distribution of such states). If this was the $k$-th visit to 
${\it Env}_i'$ then, from $\state'$, $\zstrat'$ plays
an $\epsilon\big/2^{k+1}$-optimal strategy w.r.t.\ $\game_i$
(with the same modification as above if it visits ${\it Env}_i'$ again).
We can now bound the error of $\zstrat'$ from $\state$ as follows.
The set of plays which visit ${\it Env}_i'$ infinitely often
contribute no error, since they almost surely reach the target by
$(1-\lambda)^\infty=0$.
Since all transitions are at least value-preserving in~$\game$ and hence in~$\game_i$,
the error of the plays which visit ${\it Env}_i'$ at most $j$
times is bounded by $\sum_{k=1}^j \epsilon \big/ 2^{k}$.
Therefore, the error of $\zstrat'$ from $\state$ in $\game_{i+1}$ is bounded
by $\epsilon$ and thus 
$\valueof{\game_{i+1}}{\state} = \valueof{\game_i}{\state}$.

Finally, we can construct the player~$\pz$ MD winning strategy $\hat{\zstrat}$
as the limit of the MD strategies $\zstrat_i'$, which are all compatible with
each other by the construction of the games $\game_i$.
We obtain
$\probm_{\game,\state_i,\hat{\zstrat},\ostrat}(\reach{\reachset})  >
\valueof{\game}{s_i} - 2^{-i}$ for all $i \in \N$.
Let $s \in \states$.
Since $s=s_i$ holds for infinitely many~$i$, we conclude
Thus $\probm_{\game,\state,\hat{\zstrat},\ostrat}(\reach{\reachset}) \ge
\valueof{\game}{s}$ as required. 
\end{proof}

Towards a proof of items (2) and~(3) of Theorem~\ref{thm:reachability}, we consider the operation~$\rvi(\game)$, defined before the statement of Theorem~\ref{thm:reachability}.
The following lemma shows that in reachability games
all value-increasing transitions of player~$\po$ can be removed without 
changing the value of any state (although the outcome of the threshold reachability game may change in general).


\begin{lemma}\label{lem:rvi}
Let $\game$ be a finitely branching reachability game and $\game' := \rvi(\game)$.
Then for all $\state \in \states$ we have $\valueof{\game'}{\state} =
\valueof{\game}{\state}$.
Thus $\rvi(\game')=\game'$.
\end{lemma} 
\begin{proof}
Since only $\po$ transitions are removed, we trivially have 
$\valueof{\game'}{\state} \ge \valueof{\game}{\state}$.
For the other inequality observe that the optimal minimizing strategy
of Lemma~\ref{lem:optmin} never takes any value-increasing transition
and thus also guarantees the value in $\game'$.
Thus also $\valueof{\game'}{\state} \le \valueof{\game}{\state}$.
\end{proof}
Lemma~\ref{lem:rvi} is in sharp contrast to Example~\ref{ex:exampleLimitOrdinal} on page~\pageref{ex:exampleLimitOrdinal}, which showed that the removal of value-\emph{decreasing} transitions can change the value of states and can cause further transitions to become value-decreasing.

Similar to the proof of Theorem~\ref{thm:amost-sure-strong-det}, the proof of the following lemma considers a transfinite sequence of subgames, where each subgame is obtained by removing the value-decreasing transitions from the previous subgames.
\begin{lemma} \label{lem:pre-reach-unreach-geq-MD}
Let $\game$ be a finitely branching game with reachability objective $\reach{\reachset}$.
Then there exist a player~$\pz$ MD strategy~$\hat\zstrat$ and a player~$\po$ MD strategy~$\hat\ostrat$ such that for all states $s \in \states$, if $\game = \rvi(\game)$ or $\valueof{\game}{s} = 1$, then the following is true:
\begin{equation*}
\begin{aligned}
\forall\, \ostrat \in \ostratset_\game: & \ \probm_{\game,s,\hat\zstrat,\ostrat}(\reach{\reachset}) \ge \valueof{\game}{s}  \quad \text{ or} \\
\forall\, \zstrat \in \zstratset_\game: & \ \probm_{\game,s,\zstrat,\hat\ostrat}(\reach{\reachset}) < \valueof{\game}{s}.
\end{aligned}
\end{equation*}
\end{lemma}
\begin{proof}
We construct a transfinite sequence of subgames $\game_\alpha$, 
where $\alpha \in \ord$ is an ordinal number, by stepwise removing
certain transitions.
Let $\transition_\alpha$ denote the set of transitions of the subgame~$\game_\alpha$.

First, let $\game_0 := \rvi(\game)$. Since $\game$ is assumed to have no dead
ends, it follows from the definition of $\rvi$ that $\game_0$ does not contain
any dead ends either.
In the following, we only remove transitions of player~$\pz$. The resulting
games $\game_\alpha$ with $\alpha >0$ may contain dead ends, but these are
always considered to be losing for player~$\pz$. (Formally, one might add a
dummy loop at these states.)
For each $\alpha \in \ord$ we define a set~$D_\alpha$ as the set of transitions that are controlled by player~$\pz$ and that are value-decreasing in~$\game_\alpha$.
For any $\alpha \in \ord \setminus \{0\}$ we define 
$\transition_\alpha := \transition \setminus \bigcup_{\gamma < \alpha} D_\gamma$.

Since the sequence of sets $\transition_\alpha$
is non-increasing and 
we assumed that our game $\game$ has only countably many states and
transitions, it follows that this sequence of games $\game_\alpha$
converges at some ordinal $\beta$ where $\beta \le \omega_1$ (the first
uncountable ordinal). I.e., we have $\game_\beta = \game_{\beta +1}$.
In particular there are no value-decreasing player~$\pz$ transitions
in $\game_\beta$, i.e., $D_{\beta} = \emptyset$.
 
The removal of transitions of player~$\pz$ can only decrease the value 
of states, and the operation $\rvi$ is value preserving by Lemma~\ref{lem:rvi}. 
Thus $\valueof{\game_\beta}{\state} \le \valueof{\game_\alpha}{\state} \le \valueof{\game}{\state}$ for all $\alpha \in \ord$.
We define the \emph{index} of a state $\state$ by 
$I(\state) := \min\{\alpha\in\ord \,|\, \valueof{\game_\alpha}{\state} <
\valueof{\game}{\state}\}$, 
and as $\undef$ if the set is empty.

\smallskip{\noindent\bf Strategy~$\hat\zstrat$:}
Since $\game_\beta$ does not have value-decreasing transitions, we can invoke Lemma~\ref{lem:reach-opt-uniform} to obtain a player~$\pz$ MD strategy~$\hat{\zstrat}$ with $\probm_{\game_\beta,s,\hat{\zstrat},\ostrat}(\reach{\reachset}) \ge \valueof{\game_\beta}{s} = \valueof{\game}{s}$ for all~$\pi$ and for all~$s$ with $I(s) = \undef$.
We show that, if $I(s) = \undef$ and either $\valueof{\game}{s} = 1$ or $\game = \rvi(\game)$, then also in~$\game$ we have $\probm_{\game,s,\hat{\zstrat},\ostrat}(\reach{\reachset}) \ge \valueof{\game}{s}$.
The only potential difference in the game on $\game$ is that
$\ostrat$ could take a $\po$ transition, say $s' \transition s''$, that is
present in~$\game$ but not in $\game_\beta$.
Since all $\po$ transitions of $\game_0$ are kept in $\game_\beta$,
such a transition would have been removed in the step 
$\game_0 := \rvi(\game)$.
We show that this is impossible.

For the first case suppose that $s$ satisfies $I(s) = \undef$ and $\valueof{\game}{s} = 1$.
It follows $\valueof{\game_\beta}{s} = 1$.
Since $\game_\beta$ does not have value-decreasing transitions, we have $\valueof{\game_\beta}{s'} = \valueof{\game_\beta}{s''} = 1$, hence $\valueof{\game}{s'} = \valueof{\game}{s''} = 1$, so the transition $s' \transition s''$ is not value-increasing in~$\game$.
Hence the transition is present in~$\game_0$, hence also in~$\game_\beta$.

For the second case suppose $\game = \rvi(\game)$.
Since $\game$ does not contain any value-increasing transitions, the transition $s' \transition s''$ is not value-increasing in~$\game$.
So it is present in~$\game_0$, and thus also in~$\game_\beta$.

It follows that under $\hat{\zstrat}$ the play remains in the states of
$\game_\beta$ and only uses transitions that are present in $\game_\beta$, 
regardless of the strategy $\ostrat$.
In this sense, all plays under $\hat{\zstrat}$ on~$\game$ coincide with plays on~$\game_\beta$.
Hence
$\probm_{\game,s,\hat{\zstrat},\ostrat}(\reach{\reachset}) =
\probm_{\game_\beta,s,\hat{\zstrat},\ostrat}(\reach{\reachset}) \ge \valueof{\game}{s}$.

\smallskip{\noindent\bf Strategy~$\hat\ostrat$:}
It now suffices to define a player~$\po$ MD strategy~$\hat\ostrat$ so that we have $\probm_{\game,s,\zstrat,\hat\ostrat}(\reach{\reachset}) < \valueof{\game}{s}$ for all~$\zstrat$ and for all $s$ with $I(s) \in \ord$.
This strategy~$\hat\ostrat$ is defined as follows.
\begin{itemize}
\item
If $I(\state)=\alpha$ then $\hat{\ostrat}(\state)=\state'$ where $\state'$ 
is an arbitrary but fixed successor of $\state$ 
where transition $\state\transition\state'$ is present in~$\game_\alpha$ and
$\valueof{\game_\alpha}{\state} = \valueof{\game_\alpha}{\state'}$
and $I(\state')=I(\state)=\alpha$.
This exists by the assumption that $\game$ is finitely branching
and the definition of~$\game_\alpha$.
In particular, since the transition $\state\transition\state'$ 
is present in~$\game_\alpha$, it is not value-increasing in the game $\game$;
otherwise it would have been removed in the step from $\game$ to~$\game_0$.
\item
If $I(\state)=\undef$,
$\hat{\ostrat}$ plays the optimal minimizing MD strategy on $\game$ from
Lemma~\ref{lem:optmin}, i.e., we have
$\hat{\ostrat}(\state)=\state'$ where $\state'$ 
is an arbitrary but fixed successor of $\state$ in $\game$ with
$\valueof{\game}{\state} = 
\valueof{\game}{\state'}$.
\end{itemize}
Considering both cases, it follows that strategy~$\hat{\ostrat}$
is optimal minimizing in~$\game$.

Let $s_0$ be an arbitrary state with $I(s_0) \in \ord$.
To show that $\probm_{\game,s_0,\zstrat,\hat{\ostrat}}(\reach{\reachset}) < \valueof{\game}{s_0}$ holds for all~$\zstrat$, let $\zstrat$ be any strategy of player~$\pz$.
Let $\alpha \neq \undef$ be the smallest index among the states that can be reached with positive probability from~$\state_0$ under the strategies $\zstrat,\hat{\ostrat}$.
Let $\state_1$ be such a state with index~$\alpha$.
In the following we write $\zstrat$ also for the strategy~$\zstrat$ after a partial play leading from~$\state_0$ to~$\state_1$ has been played.

Suppose that the play from~$\state_1$ under the strategies $\zstrat,\hat{\ostrat}$ 
always remains in~$\game_\alpha$.
Strategy~$\hat{\ostrat}$ might not be optimal minimizing in~$\game_\alpha$ in general.
However, we show that it is optimal minimizing in $\game_\alpha$ from all states with index
$\ge \alpha$.
Let $\state$ be a $\po$ state with index $I(\state) = \alpha' \ge \alpha$.
By definition of $\hat\ostrat$ we have
$\hat\ostrat(\state)=\state'$ where the transition $\state \transition \state'$ is present in $\game_{\alpha'}$
with $\valueof{\game_{\alpha'}}{\state} = \valueof{\game_{\alpha'}}{\state'}$
and $I(\state') = I(\state) =\alpha'$.
In the case where $\alpha' = \alpha$ this directly implies that the step
$\state \transition \state'$ is optimal minimizing in $\game_\alpha$.
The remaining case is that $\alpha' > \alpha$.
Here, by definition of the index,
$\valueof{\game}{\state} = \valueof{\game_\alpha}{\state}$ and
$\valueof{\game}{\state'} = \valueof{\game_\alpha}{\state'}$.
Since the transition $\state \transition \state'$ is present in $\game_{\alpha'}$,
it is also present in $\game_0$ and $\game_\alpha$.
Since $\game_0 = \rvi(\game)$, this transition is not value-increasing in $\game$.
Also, it is not value-decreasing in $\game$, because it is
a $\po$ transition.
Therefore $\valueof{\game}{\state} = \valueof{\game}{\state'}$,
and thus 
$\valueof{\game_\alpha}{\state} = \valueof{\game_\alpha}{\state'}$.
Also in this case the step
$\state \transition \state'$ is optimal minimizing in $\game_\alpha$.

So the only possible exceptions where strategy~$\hat{\ostrat}$ might 
not be optimal minimizing in~$\game_\alpha$ are states with index $<\alpha$.
Since we have assumed above that such states cannot be reached under $\zstrat, \hat{\ostrat}$,
it follows that $\probm_{\game,\state_1,\zstrat,\hat{\ostrat}}(\reach{\reachset})
\le \valueof{\game_\alpha}{\state_1} < \valueof{\game}{\state_1}$.

Now suppose that the play from~$\state_1$ under $\zstrat,\hat{\ostrat}$, 
with positive probability, takes a transition, say $s_2 \transition s_3$, that is not present in~$\game_\alpha$.
Then this transition was value-decreasing for some game~$\game_{\alpha'}$ with ${\alpha'}<\alpha$: that is, $\valueof{\game_{\alpha'}}{\state_2} > \valueof{\game_{\alpha'}}{\state_3}$.
Since the indices of both $\state_2$ and~$\state_3$ are $\ge \alpha >{\alpha'}$, we have
$\valueof{\game}{\state_2} = \valueof{\game_{\alpha'}}{\state_2} > \valueof{\game_{\alpha'}}{\state_3} = \valueof{\game}{\state_3}$.
Hence the transition $\state_2 \transition \state_3$ is value-decreasing in~$\game$.
Since $\hat{\ostrat}$ is optimal minimizing in~$\game$, we also have
$\probm_{\game,\state_1,\zstrat,\hat{\ostrat}}(\reach{\reachset}) < \valueof{\game}{\state_1}$.

Since $\hat{\ostrat}$ is optimal minimizing in~$\game$, we conclude that we have $\probm_{\game,\state_0,\zstrat,\hat{\ostrat}}(\reach{\reachset}) < \valueof{\game}{\state_0}$.
\end{proof}

We are now ready to prove Theorem~\ref{thm:reachability}.
\begin{proof}[Proof of Theorem~\ref{thm:reachability}]
Let $\game$ be a finitely branching game with reachability objective $\quantobj{\reach{\reachset}}{\constraint\const}$.
Let $s_0 \in \states$ be an arbitrary initial state.

Suppose $\valueof{\game}{s_0} < c$.
Then player~$\po$ wins with the MD strategy from Lemma~\ref{lem:optmin}.

Suppose $\valueof{\game}{s_0} > c$.
Let $\delta := \valueof{\game}{\state_0} - \const > 0$.
By Lemma~\ref{lem:reach-approx} there are a strategy $\zstrat \in \zstratset$ and $n \in \N$ such that $\probm_{\game,\state_0,\zstrat,\ostrat}(\reachn{n}{\reachset}) > \valueof{\game}{\state_0} - \frac{\delta}{2} > \const$ holds for all $\ostrat \in \ostratset$.
The strategy $\zstrat$ plays on the subgame~$\game'$ with state space $\states' = \{s'\in \states \mid s \transition^{\le n} s'\}$, which is finite since $\game$ is finitely branching.
Therefore, there exists an MD strategy~$\zstrat'$ with $\probm_{\game',\state_0,\zstrat',\ostrat}(\reach{\reachset}) \ge \probm_{\game,\state_0,\zstrat,\ostrat}(\reachn{n}{\reachset})$.
Since $\states' \subseteq \states$, the strategy~$\zstrat'$ also applies in~$\game$, hence $\probm_{\game,\state_0,\zstrat',\ostrat}(\reach{\reachset}) \ge \probm_{\game',\state_0,\zstrat',\ostrat}(\reach{\reachset})$.
By combining the mentioned inequalities we obtain that $\probm_{\game,\state_0,\zstrat',\ostrat}(\reach{\reachset}) > c$ holds for all $\ostrat \in \ostratset$.
So the MD strategy~$\zstrat'$ is winning for player~$\pz$.

It remains to consider the case $\valueof{\game}{s_0} = c$.
Let us discuss the four cases from the statement of Theorem~\ref{thm:reachability} individually.
\begin{enumerate}
\item[(4)] If $\mathord{\constraint} = \mathord{>}$ then player~$\po$ wins with the MD strategy from Lemma~\ref{lem:optmin}.
\end{enumerate}
So for the remaining cases it suffices to consider the threshold objective $\quantobj{\reach{\reachset}}{\ge \valueof{\game}{s_0}}$.
\begin{enumerate}
\item[(1)] If player~$\pz$ does not have value-decreasing transitions then player~$\pz$ wins with the MD strategy from Lemma~\ref{lem:reach-opt-uniform}.
\item[(2)] If player~$\po$ does not have value-increasing transitions then Lemma~\ref{lem:pre-reach-unreach-geq-MD} supplies either player~$\pz$ or player~$\po$ with an MD winning strategy.
\item[(3)] If $\const = \valueof{\game}{s_0} = 1$ then, again, Lemma~\ref{lem:pre-reach-unreach-geq-MD} supplies either player~$\pz$ or player~$\po$ with an MD winning strategy.
\end{enumerate}
This completes the proof of Theorem~\ref{thm:reachability}.
\end{proof}

\subsection{B\"uchi and co-B\"uchi Objectives}\label{sec:buchi}

Let $\formula$ be the B\"uchi objective.
(The co-B\"uchi objective is dual.)
Quantitative B\"uchi objectives $\quantobj{\formula}{\constraint\const}$
with $\const \in (0,1)$ are not strongly determined, not even for finitely
branching games (Theorem~\ref{thm:Buchi-no-strong-det}),
but positive probability $\quantobj{\formula}{>0}$ and almost-sure 
$\quantobj{\formula}{\ge 1}$ B\"uchi objectives are strongly determined
(Theorem~\ref{thm:amost-sure-strong-det}).

However, $\quantobj{\formula}{>0}$ objectives are not strongly FR-determined,
even in finitely branching systems.
Even in the special case of finitely branching MDPs (where player~$\po$ is passive
and the game is trivially strongly determined), player~$\pz$ may require
infinite memory to win \cite{Krcal:Thesis:2009}.

In infinitely branching games, the almost-sure B\"uchi objective 
$\quantobj{\formula}{\ge 1}$ is not strongly FR-determined, because it
subsumes the almost-sure reachability objective;
cf.\ Subsection~\ref{sec:reachability}.

In contrast, in finitely branching games, the almost-sure B\"uchi objective 
$\quantobj{\formula}{\ge 1}$ is strongly MD-determined, as the following theorem shows:

\begin{theorem} \label{thm:Buchi}
Let $\game$ be a finitely branching game with objective $\buchi{\reachset}$.
Then there exist a player~$\pz$ MD strategy~$\hat\zstrat$ and a player~$\po$ MD strategy~$\hat\ostrat$ such that for all states $s \in \states$:
\begin{equation*}
\begin{aligned}
\forall\, \ostrat \in \ostratset_\game: & \ \probm_{\game,s,\hat\zstrat,\ostrat}(\buchi{\reachset}) = 1  \quad \text{ or} \\
\forall\, \zstrat \in \zstratset_\game: & \ \probm_{\game,s,\zstrat,\hat\ostrat}(\buchi{\reachset}) < 1.
\end{aligned}
\end{equation*}
Hence finitely branching almost-sure B\"uchi games are strongly MD-determined.
\end{theorem}

For the proof we need the following lemmas, which are variants of Lemmas \ref{lem:optmin} and~\ref{lem:reach-opt-uniform} for the objective $\reachp{\reachset}$, which is defined as:
\[
\reachp{\reachset} := \{\state_0\state_1\cdots \in S^{\omega}\mid \exists\, i \ge 1.\,
\state_i \in \reachset\}
\]
The difference to~$\reach{\reachset}$ is that $\reachp{\reachset}$ requires a path to~$\reachset$ that involves at least one transition.

\begin{lemma} \label{reachp-opt-min}
Let $\game$ be a finitely branching game with objective $\reachp{\reachset}$.
Then there is an MD strategy $\ostrat \in \ostratset$ that is optimal minimizing in every state.
\end{lemma}
\begin{proof}
Outside~$\reachset$, the objectives $\reach{\reachset}$ and $\reachp{\reachset}$ coincide, so outside~$\reachset$, the MD strategy~$\ostrat$ from Lemma~\ref{lem:optmin} is optimal minimizing for $\reachp{\reachset}$.
Any $s \in \reachset \cap \ostates$ with $\valueof{\game}{s} < 1$ must have a transition $s \transition s'$ with $s' \notin \reachset$ and $\valueof{\game}{s} = \valueof{\game}{s'}$, where the value is always meant with respect to~$\reachp{\reachset}$.
Set $\ostrat(s) := s'$.
Then $\ostrat$ is optimal minimizing in every state, as desired.
\end{proof}
\begin{lemma} \label{reachp-opt-max}
Let $\game$ be a finitely branching game with objective $\reachp{\reachset}$.
Suppose player~$\pz$ does not have value-decreasing transitions.
Then there is an MD strategy $\zstrat \in \zstratset$ that is optimal maximizing in every state.
\end{lemma}
\begin{proof}
Outside~$\reachset$, the objectives $\reach{\reachset}$ and $\reachp{\reachset}$ coincide, so outside~$\reachset$, the MD strategy~$\zstrat$ from Lemma~\ref{lem:reach-opt-uniform} is optimal maximizing for $\reachp{\reachset}$.
Any $s \in \reachset \cap \zstates$ must have a transition $s \transition s'$ with $s' \in \reachset$ or $\valueof{\game}{s} = \valueof{\game}{s'}$, where the value is always meant with respect to~$\reachp{\reachset}$.
Set $\zstrat(s) := s'$.
Then $\zstrat$ is optimal maximizing in every state, as desired.
\end{proof}

With this at hand, we prove Theorem~\ref{thm:Buchi}.
\begin{proof}[Proof of Theorem~\ref{thm:Buchi}]
We proceed similarly to the proof of Theorem~\ref{thm:amost-sure-strong-det}.
In the present proof, whenever we write $\valueof{\game'}{s}$ for a subgame~$\game'$ of~$\game$, we mean the value of state~$s$ with respect to $\reachp{\reachset \cap \states'}$, where $\states' \subseteq \states$ is the state space of~$\game'$.

In order to characterize the winning sets of the players with respect to the objective $\buchi{\reachset}$, we construct a transfinite sequence of subgames $\game_\alpha$ of~$\game$, 
where $\alpha \in \ord$ is an ordinal number, by stepwise removing
certain states, along with their incoming transitions.
Let $\states_\alpha$ denote the state space of the subgame~$\game_\alpha$.
We start with $\game_0 := \game$.
Given~$\game_\alpha$, define~$D_\alpha^0$ as the set of states $s \in \states_\alpha$ with $\valueof{\game_\alpha}{s} < 1$, and for any $i \ge 0$ define $D_\alpha^{i+1}$ as the set of states $s \in \big(\states_\alpha \setminus \bigcup_{j=0}^i D_\alpha^j\big) \cap (\ostates \cup \rstates)$ that have a transition $s \transition s'$ with $s' \in D_\alpha^i$.
The set $\bigcup_{i\in \N} D_\alpha^i$ can be seen as the backward closure of~$D_\alpha^0$ under random transitions and transitions controlled by player~$\po$.
For any $\alpha \in \ord\setminus\{0\}$ we define 
$\states_\alpha := \states \setminus \bigcup_{\gamma < \alpha} \bigcup_{i \in \N} D_\gamma^i$.

Since the number of states never increases and $\states$ is countable, it follows that this sequence of games $\game_\alpha$ converges at some ordinal~$\beta$ where $\beta \le \omega_1$ (the first uncountable ordinal).
That is, we have $\game_\beta = \game_{\beta +1}$.

As in the proof of Theorem~\ref{thm:amost-sure-strong-det}, some games
$\game_\alpha$ may contain dead ends, which are always considered to be losing
for player~$\pz$. However, $\game_\beta$ does not contain dead ends.
(If $\states_\beta$ is empty then player~$\pz$ loses.)
We define the \emph{index}, $I(\state)$, of a state $\state$ as the ordinal~$\alpha$ with $s \in \bigcup_{i \in \N} D_\alpha^i$, and as $\undef$ if such an ordinal does not exist.
For all states~$s \in \states$ we have:
\[
I(s) = \undef \ \Leftrightarrow \ s \in \states_\beta \ \Leftrightarrow \ \valueof{\game_\beta}{s} = 1
\]
In particular, player~$\pz$ does not have value-decreasing transitions in~$\game_\beta$.
We show that states~$s$ with $I(s) \in \ord$ are in ${\owinset{\buchi{\reachset}}{<1}}_{\!\!\game}$, and states~$s$ with $I(s) = \undef$ are in ${\zwinset{\buchi{\reachset}}{=1}}_{\!\!\game}$, and in each case we give the claimed witnessing MD strategy.

\smallskip{\noindent\bf Strategy~$\hat\ostrat$:}
We define the claimed MD strategy~$\hat\ostrat$ for all $s \in \ostates$ with $I(s) = \alpha \in \ord$ as follows.
For all $s \in D_\alpha^0$, define~$\hat\ostrat(s)$ as in the MD strategy from Lemma~\ref{reachp-opt-min} for $\game_\alpha$ and~$\reachp{\reachset \cap \states_\alpha}$.
For all $s \in D_\alpha^{i+1} \cap \ostates$ for some $i \in \N$, define $\hat\ostrat(s) := s'$ such that $s \transition s'$ and $s' \in D_\alpha^i$.

In each~$\game_\alpha$, strategy~$\hat\ostrat$ coincides with the strategy from Lemma~\ref{reachp-opt-min}, except possibly in states~$s \in \states_\alpha$ with $\valueof{\game_\alpha}{s} = 1$.
It follows that $\hat\ostrat$ is optimal minimizing for all $\game_\alpha$ with $\alpha \in \ord$.

We show by transfinite induction on the index that
$\probm_{\game,s,\zstrat,\hat\ostrat}(\buchi{\reachset}) < 1$ holds for all
states $s \in \states$ with $I(s) \in \ord$ and for all player~$\pz$ strategies~$\zstrat$.
For the induction hypothesis, let $\alpha$ be an ordinal for which this holds for all states~$s$ with $I(s) < \alpha$.
For the inductive step, let $s \in \states$ be a state with $I(s) = \alpha$, and let $\zstrat$ be an arbitrary player~$\pz$ strategy in~$\game$.

\begin{itemize}
\item
Let $s \in D_\alpha^0$.
Suppose that the play from~$s$ under the strategies $\zstrat,\hat{\ostrat}$ always remains in~$\states_\alpha$, i.e., the probability of ever leaving~$\states_\alpha$ under $\zstrat,\hat{\ostrat}$ is zero.
Then any play in~$\game$ under these strategies coincides with a play in~$\game_\alpha$, so we have $\probm_{\game,s,\zstrat,\hat\ostrat}(\reachp{\reachset}) = \probm_{\game_\alpha,s,\zstrat,\hat\ostrat}(\reachp{\reachset \cap \states_\alpha})$.
Since $\hat\ostrat$ is optimal minimizing in~$\game_\alpha$, we have $\probm_{\game_\alpha,s,\zstrat,\hat\ostrat}(\reachp{\reachset \cap \states_\alpha}) \le \valueof{\game_\alpha}{s} < 1$.
Since $\buchi{\reachset} \subseteq \reachp{\reachset}$, we have $\probm_{\game,s,\zstrat,\hat\ostrat}(\buchi{\reachset}) \le \probm_{\game,s,\zstrat,\hat\ostrat}(\reachp{\reachset})$.
By combining the mentioned equalities and inequalities we get 
$\probm_{\game,s,\zstrat,\hat\ostrat}(\buchi{\reachset}) < 1$, as desired.

Now suppose otherwise, i.e., the play from~$s$ under $\zstrat,\hat{\ostrat}$, 
with positive probability, enters a state $s' \notin \states_\alpha$, hence $I(s') < \alpha$.
By the induction hypothesis we have $\probm_{\game,s',\zstrat',\hat\ostrat}(\buchi{\reachset}) < 1$ for any~$\zstrat'$.
Since the probability of entering~$s'$ is positive, we conclude $\probm_{\game,s,\zstrat,\hat\ostrat}(\buchi{\reachset}) < 1$, as desired.
\item
Let $s \in D_\alpha^{i}$ for some $i \ge 1$.
It follows from the definitions of~$D_\alpha^{i}$ and of~$\hat\ostrat$ that $\hat\ostrat$ induces a partial play of length $i+1$ from $s$ to a state $s' \in D_\alpha^0$ (player~$\pz$ does not play on this partial play).
We have shown above that $\probm_{\game,s',\zstrat,\hat\ostrat}(\buchi{\reachset}) < 1$.
It follows that $\probm_{\game,s,\zstrat,\hat\ostrat}(\buchi{\reachset}) < 1$, as desired.
\end{itemize}
We conclude that we have $\probm_{\game,s,\zstrat,\hat\ostrat}(\buchi{\reachset}) < 1$ for all $\zstrat$ and all $s \in \states$ with $I(s) \in \ord$.

\smallskip{\noindent\bf Strategy~$\hat\zstrat$:}
We define the claimed MD strategy~$\hat\zstrat$ for all $s \in \zstates$ with $I(s) = \undef$ to be the MD strategy from Lemma~\ref{reachp-opt-max} for~$\game_\beta$ and~$\reachp{\reachset \cap \states_\beta}$.
This definition ensures that player~$\pz$ never takes a transition in~$\game$ that leaves~$\states_\beta$.
Random transitions and player~$\po$ transitions in~$\game$ never leave~$\states_\beta$ either:
indeed, if $s' \in \states$ with $I(s') = \alpha \in \ord$ then $s' \in D_\alpha^{i}$ for some~$i$, hence if $s \in \ostates \cup \rstates$ and $s \transition s'$ then $I(s) \le \alpha$.
We conclude that starting from~$\states_\beta$ all plays in~$\game$ remain in~$\states_\beta$, under $\hat\zstrat$ and all player~$\po$ strategies.

Let $s \in \states_\beta$, hence $\valueof{\game_\beta}{s} = 1$.
Let $\ostrat$ be any player~$\po$ strategy.
Since $\hat\zstrat$ is optimal maximizing in~$\game_\beta$, we have $\probm_{\game_\beta,s,\hat\zstrat,\ostrat}(\reachp{\reachset \cap \states_\beta}) = 1$.
As argued above, $\states_\beta$ is not left even in~$\game$, hence $\probm_{\game,s,\hat\zstrat,\ostrat}(\reachp{\reachset \cap \states_\beta}) = 1$.

Therefore $\probm_{\game,s,\hat\zstrat,\ostrat}(\reachp{\reachset \cap \states_\beta}) = 1$ holds for all $s \in \states_\beta$ and all~$\ostrat$.
Since B\"uchi is repeated reachability, we also have $\probm_{\game,s,\hat\zstrat,\ostrat}(\buchi{\reachset}) = 1$ for all $\ostrat$ and all $s \in \states$ with $I(s) = \undef$.
\end{proof}

\section{Conclusions and Open Problems}\label{sec:conclusion}
With the results of this paper at hand, let us review the landscape of strong determinacy for stochastic games.
We have shown that almost-sure objectives are strongly determined (Theorem~\ref{thm:amost-sure-strong-det}), even in the infinitely branching case.

Let us review the finitely branching case. 
Quantitative reachability games are strongly determined \cite{Krcal:Thesis:2009,BBKO:IC2011,Brozek:TCS2013}.
They are generally not strongly FR-determined~\cite{Kucherabook11}, but they are strongly MD-determined under any of the conditions provided by Theorem~\ref{thm:reachability}.
Almost-sure reachability games and even almost-sure B\"uchi games are strongly MD-determined (Theorems \ref{thm:reachability} and~\ref{thm:Buchi}).
Almost-sure co-B\"uchi games are generally not strongly
FR-determined~\cite{Krcal:Thesis:2009}, even if player~$\pz$ is passive,
because player~$\po$ may need infinite memory to win.
However, the following question is open: if a state is almost-surely winning for
player~$\pz$ in a co-B\"uchi game, 
does player~$\pz$ also have a winning MD strategy?

The same question is open for infinitely branching almost-sure reachability games (these games are generally not strongly FR-determined either~\cite{Kucherabook11}).
In fact, one can show that a positive answer to the former question implies a positive answer to the latter question.

{\smallskip\bf\noindent Acknowledgements.}
This work was partially supported by the EPSRC
through grants EP/M027287/1, EP/M027651/1, EP/P020909/1 and EP/M003795/1
and by St. John's College, Oxford.


\end{document}